\documentclass[11pt, onecolumn, draftnocls]{IEEEtran}
\IEEEoverridecommandlockouts
\usepackage{amsfonts,booktabs}
\usepackage{amsmath}
\usepackage{graphicx}
\usepackage{subfigure}
\usepackage{times}
\usepackage{geometry}
\usepackage{setspace}
\usepackage{pgfplots}
\usepackage{pgfplotstable}
\usepackage[subnum]{cases}
\usepackage{filecontents}
\usepackage{multirow}
\usepackage{color}
\usepackage{pdfsync}
\let\oldbibliography\thebibliography
\renewcommand{\thebibliography}[1]{%
  \oldbibliography{#1}%
  \setlength{\itemsep}{0pt}%
}

\newtheorem{sideremark}{Remark}

\newtheorem{definition}         {Definition}[section]
\newtheorem{lemma}{Lemma}
\newtheorem{theorem}{Theorem}

\def\reals{{\rm I\kern-.17em R}}
\def\nats{{\rm I\kern-.17em N}}

\newcommand{\beq}{\begin{equation}}
\newcommand{\eeq}{\end{equation}}
\newcommand{\beqa}{\begin{eqnarray}}
\newcommand{\eeqa}{\end{eqnarray}}

\newcommand{\mathbi}[1]{\ensuremath \textbf{\em #1}}
\newcommand{\paren}[1]{\left(#1\right)}
\newcommand{\sqparen}[1]{\left[#1\right]}
\newcommand{\brparen}[1]{\left\{#1\right\}}
\newcommand{\field}[1]{\ensuremath{\mathbb{#1}}}
\newcommand{\abs}[1]{\left|#1\right|} 
\newcommand{\N}{\ensuremath{\field{N}}} 
\newcommand{\R}{\ensuremath{\field{R}}} 
\newcommand{\Rp}{\ensuremath{\R_+}} 
\newcommand{\I}[1]{\ensuremath{\mathsf{1}_{\left\{#1\right\}}}} 
\newcommand{\Inb}[1]{\ensuremath{\mathsf{1}_{#1}}} 
\newcommand{\ra}{\ensuremath{\rightarrow}} 
\newcommand{\PR}[1]{\ensuremath{\mathsf{Pr}\left\{#1\right\}}} 
\newcommand{\PRP}[1]{\ensuremath{\mathsf{Pr}\left(#1\right)}} 
\newcommand{\EW}{\ensuremath{\mathsf{E}}} 
\newcommand{\ES}[1]{\ensuremath{\mathsf{E}\left[#1 \right]}} 
\newcommand{\e}[1]{\ensuremath{{\rm e}^{#1}}} 

\newcommand{\LO}[1]{\ensuremath{o\paren{#1}}}

\renewcommand{\vec}[1]{\ensuremath{\boldsymbol{#1}}}

\newcommand{\logp}[1]{\ensuremath{\log\paren{#1}}}

\newcommand{\CRTPIL}{\ensuremath{\frac{h_i}{\lambda_N+\mu_Ng_i}}}

\newcommand{\CRDTPIL}{\ensuremath{F_{\lambda_N,\mu_N}\paren{x}}}
\newcommand{\ICRDTPIL}[1]{\ensuremath{F^{-1}_{\lambda_N,\mu_N}\paren{1-#1}}}
\newcommand{\PTPIL}{\ensuremath{P_{\rm DTPIL}\paren{h_i,g_i}}}
\newcommand{\Xs}[2]{\ensuremath{X_N^\star\paren{#1,#2}}}
\newcommand{\Xd}[2]{\ensuremath{X_N^\diamond\paren{#1,#2}}}
\newcommand{\RTPIL}[1]{\ensuremath{R^\star_{\rm DTPIL}\paren{#1,N}}}
\newcommand{\Xt}{\ensuremath{\tilde{X}_N}}

\newcommand{\CRIL}{\ensuremath{\frac{h_i}{g_i}}}
\newcommand{\Ys}{\ensuremath{Y_N^\star}}
\newcommand{\Yt}{\ensuremath{\tilde{Y}_N}}
\newcommand{\Yd}{\ensuremath{Y_N^\diamond}}
\newcommand{\PIL}{\ensuremath{P_{\rm DIL}\paren{h_i,g_i}}}
\newcommand{\ICRDIL}[1]{\ensuremath{F^{-1}_{\frac{h}{g}}\paren{1-#1}}}
\newcommand{\RIL}[1]{\ensuremath{R^\star_{\rm DIL}\paren{#1,N}}}

\newcommand{\pfgsetting}{\pgfplotsset{
width= 8cm, 
every axis/.append style={line width=1.2pt},
label style={font=\bf\scriptsize}, 
ylabel style={yshift=-0.8em},
xlabel={Number of Secondary Users},
ylabel={Throughput (bits/sec/Hz)},
title style={font=\bf\scriptsize}, 
tick label style={font=\scriptsize,/pgf/number format/1000 sep={} },
tick style={ line width=1.5pt},
legend style={font=\bf\tiny,cells={anchor=west},solid},
every mark/.append style={solid}
}}
\newcommand{\axissetting}{ xmin=0,xmax=1000,xtick={0 ,100,200,300,400,500,600,700,800,900,1000},grid=major, tick style={color=black, major tick length={0.10 cm}}, grid=none 
}

\begin{document}

\title{Power Control and Multiuser Diversity  for the Distributed Cognitive Uplink}

\author{Ehsan Nekouei, \IEEEmembership{Student Member, IEEE}, Hazer Inaltekin, \IEEEmembership{Member, IEEE} and Subhrakanti Dey, \IEEEmembership{Senior Member, IEEE} 
\thanks{E. Nekouei and S. Dey are with the Department of Electrical and Electronic Engineering, The University of Melbourne, VIC 3010, Australia. (e-mails: e.nekouei@pgrad.unimelb.edu.au and sdey@unimelb.edu.au) 

H. Inaltekin is with the Department of Electrical and Electronics Engineering, Antalya International University, Antalya, Turkey. (e-mail: hazeri@antalya.edu.tr).}}
\maketitle
\thispagestyle{empty}
\vspace{-1cm}
\begin{abstract}
This paper studies optimum power control and sum-rate scaling laws for the distributed cognitive uplink. 
It is first shown that the optimum distributed power control policy is in the form of a threshold based water-filling power control.  Each secondary user executes the derived power control policy in a distributed fashion by using local knowledge of its direct and interference channel gains such that the resulting aggregate (average) interference does not disrupt primary's communication.  
Then, the tight sum-rate scaling laws are derived as a function of the number of secondary users $N$ under the optimum distributed power control policy.  The fading models considered to derive sum-rate scaling laws are general enough to include Rayleigh, Rician and Nakagami fading models as special cases.  When transmissions of secondary users are limited by both transmission and interference power constraints, it is shown that the secondary network sum-rate scales according to $\frac{1}{\e{}n_h}\log\logp{N}$, where $n_h$ is a parameter obtained from the distribution of direct channel power gains. For the case of transmissions limited only by interference constraints, on the other hand, the secondary network sum-rate scales according to $\frac{1}{\e{}\gamma_g}\logp{N}$, where $\gamma_g$ is a parameter obtained from the distribution of interference channel power gains.  These results indicate that the distributed cognitive uplink is able to achieve throughput scaling behavior similar to that of the centralized cognitive uplink up to a pre-log multiplier $\frac{1}{\e{}}$, whilst primary's quality-of-service requirements are met. The factor $\frac{1}{\e{}}$ can be interpreted as the cost of distributed implementation of the cognitive uplink.
\end{abstract}

\section{Introduction}
\subsection{Background and Motivation}
Cognitive radio technology has recently emerged as an aspirant solution for the problem of spectrum scarcity \cite{Goldsmith09}-\cite{Akyildiz06}. Unlike the traditional static command-and-control approach, it provides a more dynamic means for spectrum management and utilization. More specifically, cognitive radio protocols such as those in IEEE 802.22 allow the cognitive users, alternatively called secondary users (SUs), to dynamically share the underutilized frequency bands with primary users (PUs) both in time and space under various forms of primary quality-of-service (QoS) protections \cite{Mitola99b, Haykin05}. In practice, channel state information (CSI) is one of the main requisites for successful implementation of such dynamic cognitive radio protocols. However, its availability is often sidelined in most previous works \cite{Ghasemi07}-\cite{cogmud_alitajer10} by either assuming a centralized band manager or perfect instantaneous CSI feedback between primary and secondary networks. 

For example, both interference management and resource allocation tasks in cognitive radio networks heavily depend on the availibility of CSI at the secondary network. This requirement is especially more pronounced for multiuser cognitive radio networks. The assumption of the existence of a {\em centralized} entity having {\em global} knowledge of CSI may not be realistic in some certain multiuser cognitive communication scenarios, depending on the physical characteristics of wireless channels, infrastructure limitations, number of SUs and etc. In these cases, distributed utilization of CSI is the key for successful implementation of cognitive radio protocols.  To this end, the current paper explores the design of optimum distributed power control mechanisms for the cognitive uplink, allowing each SU to adjust its transmission power level based \emph{only}  on local knowledge of its CSI.  It also investigates multiuser diversity gains for the distributed cognitive uplink by deriving tight sum-rate capacity scaling laws under the optimum distributed power control mechanisms. 

In  the centralized uplink, the secondary base-station (SBS) is mainly responsible for the power control task {\em e.g.,} see \cite{RZhang09}-\cite{IH12}. That is, it first acquires global knowledge of direct (from SUs to the SBS) and interference  (from SUs to PUs) channel gains via a feedback mechanism, and then exploits this knowledge to obtain the optimum transmission power level for each SU by respecting primary QoS requirements.  Finally, the allocated transmission power levels are broadcasted to SUs by the SBS at each fading block. 
Although required for the centralized power control at the cognitive uplink per above discussion, the assumption of availability of direct and interference channel gains at the SBS within channel coherence time is often too restricting for practical cognitive multiple access networks consisting of large numbers of SUs.\footnote{The terms \emph{cognitive uplink} and \emph{cognitive multiple access network} are used interchangeably throughout the paper.}  On the other hand, unlike the centralized operation, SUs only need to have local access to their direct and interference channel gains in the distributed operating mode.  Further, it is easy for each SU to obtain local knowledge of its direct and interference channel gains using pilot signals transmitted periodically by the SBS and  primary base-station (PBS). These observations motivate the current paper, and lead to the following research questions of interest here: $(i)$ what is the structure of optimum distributed power control mechanisms for the cognitive uplink?, 
and $(ii)$ what are the fundamental throughput scaling laws of such decentralized cognitive multiple access networks under optimum power control subject to various forms of power and interference constraints?  



This paper provides important insights into these questions by studying the optimum distributed power control mechanisms for the cognitive uplink that enable SUs to accomplish the power control task in a distributed fashion while the interference at the primary network is successfully regulated.  More importantly, we evaluate the performance of our distributed power control mechanisms, in terms of the secondary network sum-rate, as the number of SUs becomes large. Our results signify the fact that distributed cognitive multiple access networks are capable of achieving throughput scaling behavior similar to that of centralized cognitive multiple access networks.  

\subsection{Contributions}

This paper has two main contributions to the cognitive radio literature. First, we derive the structure of the optimum distributed power control policy, maximizing the secondary network sum-rate for two network types: \emph{ (i)} distributed total power and interference limited (DTPIL) networks and \emph{(ii)} distributed interference limited (DIL) networks. In DTPIL networks, transmission powers of SUs are limited by a constraint on the average total transmission power of SUs and a constraint on the average interference power at a PBS. To confine the collision level, a transmission probability constraint is also considered for each SU. In DIL networks, transmission powers of SUs are limited by a constraint on the average interference power at the PBS as well as transmission probability constraints. For each network type, we show that the optimum distributed power control policy is in the form of a {\em threshold} based water-filling power control with changing water levels. 

Secondly, we study the sum-rate scaling behavior of DTPIL and DIL networks, under the optimum distributed power control policy, when distributions of direct and interference channel gains belong to a fairly large class of distribution functions called class-$\mathcal{C}$ distributions. In DTPIL networks, it is shown that the secondary network throughput scales according to $\frac{1}{\e{}n_h}\log\logp{N}$ when the transmission probability is set to $\frac{1}{N}$ for all SUs. Here, $N$ is the number of SUs, and $n_h$ is a parameter obtained from the distribution of direct channel power gains. The choice of transmission probability adds an extra dimension to the optimization problems studied in this paper. To this end, we show that although $\frac{1}{N}$ may not be the optimum transmission probability selection for  the secondary network sum-rate maximization for finite values of $N$, it is {\em asymptotically} optimum in the sense that the same throughput scaling behavior holds even under the optimum transmission probability selection. 

Analogous results are also obtained for DIL networks. In particular, it is shown that the secondary network sum-rate scales according to $\frac{1}{\e{}\gamma_g}\logp{N}$ when the transmission probability is set to $\frac{1}{N}$ for all SUs, and the optimum distributed power control policy is employed. $\gamma_g$ is a parameter obtained from the distribution of interference channel power gains. It is also shown that $\frac{1}{N}$ is the {\em asymptotically} optimum transmission probability selection for DIL networks, too.  From an engineering point of view, these results indicate that the optimum distributed power control at the cognitive uplink is capable of achieving aggregate data rates similar to those achieved through a centralized scheduler up to a pre-log multiplier $\frac{1}{\e{}}$ \cite{NIDSubmitted}. Here, $\frac{1}{\e{}}$ has the economic interpretation of the cost of avoiding feedback signals between primary and secondary networks. Our main results are summarized in Table \ref{Table: Main Results}.
\subsection{A Note on the Notation and Organization of the Paper}
In what follows, a wireless channel is said to be a Rayleigh fading channel if the channel magnitude gain is Rayleigh distributed, or equivalently the channel power gain is exponentially distributed. It is said to be Rician-$K$ fading channel if the channel magnitude gain is Rician distributed with a Rician factor $K$. By a Nakagami-$m$ distributed wireless fading channel, we mean the channel magnitude gain is Nakagami-$m$ distributed, or equivalently the channel power gain is Gamma distributed. Finally, a wireless fading channel is said to be Weibull-$c$ distributed if the channel magnitude gain is Weibull distributed with a Weibull parameter $c$.\footnote{The definition of the $c$ parameter for Weibull fading channels is adapted from \cite{Simon-Alouini05}.} Interested readers are referred to \cite{Simon-Alouini05}, \cite{Stuber96} and \cite{Bertoni88} for more details regarding fading distributions.
   
The rest of the paper is organized as follows. In the next section, we discuss relevant literature. Section \ref{Sec: SM} describes our system model and modeling assumptions. Section \ref{Sec: R&D} derives the optimum distributed power control policies and their corresponding sum-rate scaling laws for DTPIL and DIL networks. Section \ref{Sec: NR} presents our numerical studies. Section \ref{sec: conclusion} concludes the paper.
  \begin{table}[!t]
\begin{minipage}{\textwidth}
\renewcommand{\arraystretch}{1.3}
\caption{Throughput Scaling behavior of distributed cognitive radio networks}
\centering
\begin{tabular}{lll}
\toprule
\multicolumn{1}{c}{\multirow{2}{*}{ Network Model}} & \multicolumn{2}{c}{ Transmission probability}\\
 \cmidrule(r){2-3}

& \multicolumn{1}{c}{$p_N=\frac{1}{N}$} & \multicolumn{1}{c}{$p^\star_N$\footnote{$p^\star_N$ is the optimum transmission probability.}}  \\
\midrule
 Distributed Total Power And Interference Limited & $\lim\limits_{N\ra\infty}\frac{R_N\footnote{$R_N$ is the secondary network sum-rate under the optimum distributed power control policy.}}{\log\logp{N}}=\frac{1}{\e{}n_h}\footnote{ $n_h$ is parameter determined from the asymptotic tail behavior of the distribution of direct channel power gains.}$ & $\lim\limits_{N\ra\infty}\frac{R_N}{\log\logp{N}}=\frac{1}{\e{}n_h}$   \\
\midrule
 Distributed Interference Limited & $\lim\limits_{N\ra\infty}\frac{R_N}{\logp{N}}=\frac{1}{\e{}\gamma_g}\footnote{$\gamma_g$ is a parameter determined from the behavior of the distribution of interference channel power gains around the origin.}$ & $\lim\limits_{N\ra\infty}\frac{R_N}{\logp{N}}=\frac{1}{\e{}\gamma_g}$ \\
\bottomrule
\label{Table: Main Results}
\end{tabular}
\end{minipage}
\end{table}
\section{Related Work}\label{Sec: RW}
This section briefly reviews the papers that are most relevant to ours.  In this paper, we are mainly motivated by exploiting distributed techniques for optimum resource/power allocation in the cognitive uplink and the corresponding sum-rate capacity scaling via multiuser diversity. 

Optimum allocation of transmission powers in a cognitive radio setup has recently been investigated in \cite{Ghasemi07}-\cite{IH12}.  In \cite{Ghasemi07}, Ghasemi and Sousa showed that the optimum power control maximizing the ergodic capacity of a point-to-point cognitive radio link under average interference power constraint is in the form of a water-filling power control policy with changing water levels.  In \cite{RZhang09}, this result was extended to the cognitive uplink.  Particularly, they showed that, under average transmission power and average interference power constraints, the optimum power allocation policy for a cognitive uplink is in the form of an opportunistic water-filling power allocation policy. That is, the SBS schedules the SU with the best joint direct and interference channel state, and the scheduled SU employs a water-filling power allocation policy for its transmission.  

Similar results have also been obtained by considering  total power and partial CSI constraints in \cite{NID12} and \cite{NIDSubmitted}. 
 In \cite{IH12}, Inaltekin and Hanly established the binary structure of the optimum power control for the cognitive uplink operating under interference limitations without successive interference cancellation, \emph{i.e.}, see Section VI of \cite{IH12}.  They showed that the set of transmitting  SUs always corresponds to the ones having better joint channel states.  Although the single-user decoding assumption in \cite{IH12} simplifies the decoder, it complicates the power optimization problem.  The resulting optimization problem, in contrast to the one in \cite{Ghasemi07}-\cite{NIDSubmitted}, is no longer convex. 
 
This paper differs from above previous work in two important aspects.  Firstly, we focus on the distributed cognitive uplink in this paper, whereas \cite{Ghasemi07}-\cite{IH12} analyzed the centralized power control with perfect or partial CSI at the SBS.  The distributed operation requires contention control by constraining channel access probabilities, which in turn makes the studied power optimization problem here non-convex.  Secondly, similar to these previous work, our analysis in this paper starts with the consideration of optimum power control policies. However, different from them, we also investigate multiuser diversity gains in the distributed cognitive uplink as a function of the number of SUs. 

Multiuser diversity gains for the centralized cognitive radio networks with global knowledge of CSI at the SBS have also been studied in the literature extensively, {\em e.g.,} see \cite{NID12}, \cite{cogmud_twban09}-\cite{Yang12-IT}, under various types of constraints on the transmission powers of SUs. In \cite{cogmud_twban09}, the authors established logarithmic and double-logarithmic throughput scaling behavior of the cognitive uplink for Rayleigh fading channels under joint peak transmission and interference power constraints. These results were extended to cognitive multiple access, cognitive broadcast and cognitive parallel access channels in \cite{cogmid_zhang10}.  The authors in \cite{NID12}, different from \cite{cogmud_twban09} and \cite{cogmid_zhang10}, considered average power limitations (both transmission and interference), and obtained parallel ergodic sum-rate scaling results for cognitive multiple access networks under optimum power control. 

 The main point of difference between this paper and above previous work is the utilization of more practical distributed approaches for the cognitive uplink here.  Specifically, different from them, SUs in our setup independently decide to transmit (with power control) based on local knowledge of their CSI. This provides a more practical framework to study the multiuser diversity gain in the cognitive uplink, but at the expense of a more complicated optimum power control analysis ({\em e.g.,} see Appendix \ref{App: OPA-TPIL}) and the corresponding estimates on the tails of joint channel states ({\em e.g.,} see Appendices \ref{app: TPIL} and \ref{app: IL}).

In \cite{Yang12-TW},  the scheduling gain in a cognitive uplink was considered for a hybrid scheduling policy under peak transmission and interference power constraints. All SUs transmit with the same fixed power level. Under this setup, it was shown that the secondary network throughput scales logarithmically (as a function of the number of SUs) with a pre-log factor depending on the number of PUs. Similar results were extended to cognitive radio networks with multiple antennas at the SBS and PBS in \cite{Yang12-IT}.  They showed that the secondary network throughput scales logarithmically with a pre-log factor depending on the operating modes ({\em i.e,} multiple access versus broadcast) and the number of antennas at the SBS and PBS.

Other related work also includes secondary network capacity scaling in a multi-band setup such as \cite{cogmudmsd_10} and \cite{cogmud_alitajer10}. In \cite{cogmudmsd_10}, Wang {\em et al.} studied the multiuser and multi-spectrum diversity gains for a cognitive broadcast network sharing multiple orthogonal frequency bands with a primary network.  Assuming Rayleigh fading channels, they analytically derived capacity expressions for the secondary network when the transmission power at each band is limited by a constraint on the peak interference power at the primary network.  For a similar setup in \cite{cogmudmsd_10}, the authors in \cite{cogmud_alitajer10} considered $N$ secondary transmitter-receiver pairs sharing $M$ frequency bands with a primary network. Under the optimum matching of SUs with primary frequency bands, they derived a double-logarithmic scaling law for the secondary network capacity for Rayleigh fading channels. They also considered a contention-free distributed scheduling algorithm in which SUs decide to transmit (without any power control) if their received signal-to-interference-plus-noise-ratio in a frequency band is greater than a threshold level. 

Unlike \cite{Yang12-TW}-\cite{cogmud_alitajer10}, this paper considers general fading models including Rayleigh fading as a special case ({\em i.e.,} see Table \ref{F1}).  Further, all sum-rate scaling laws are derived for the contention-limited distributed cognitive uplink under optimum allocation of transmission powers to SUs, rather than assuming fixed transmission power levels as in \cite{Yang12-TW}-\cite{cogmud_alitajer10}.  The distributed power control mechanisms are designed as such they provide stringent QoS guarantees for the primary network under a collision channel model.  Hence, some parts of our analysis in this paper are expected to find greater applicability to extend multiuser diversity results derived for multi-band and multi-antenna networks in \cite{Yang12-TW}-\cite{cogmud_alitajer10} to fading models beyond Rayleigh fading and to more practical distributed communication scenarios with optimum resource allocation.   

Finally, it is important to note that multiuser diversity gains in primary multiple access networks were also studied in the literature, {\em e.g.,} see \cite{Berry06}. However, these results are not applicable to the cognitive uplink as they do not account for the impact of SUs' transmissions on the primary's QoS. What is needed in a cognitive setup is a more advanced distributed power management mechanism that can harvest multiuser diversity gains in both direct and interference channels simultaneously, whilst respecting primary's QoS requirements. This often results in solving non-convex optimization problems as in Appendix \ref{App: OPA-TPIL}, and using more complicated techniques to obtain tail estimates of joint channel states as in Appendices \ref{app: TPIL} and \ref{app: IL}.  

 
 \section{System Model}\label{Sec: SM}
We consider a cognitive uplink in which $N$ SUs communicate with an SBS and simultaneously cause interference to a PBS as depicted in Fig. \ref{F1}. Let $h_i$ and $g_i$ represent the $i$th direct and interference channel power gains, respectively. We consider the classical ergodic block fading model \cite{Tse} to model the statistical variations of all direct and interference channel gains. $\left\{h_i\right\}_{i=1}^N$ and $\left\{g_i\right\}_{i=1}^N$ are assumed to be collections of i.i.d. random variables distributed according to distribution functions $F_h\paren{x}$ and $F_g\paren{x}$, respectively. The random vectors ${\mathbi h}=\sqparen{h_1,h_2,\ldots,h_N}^\top$ and ${\mathbi g}= \sqparen{g_1,g_2,\ldots,g_N}^\top$ are assumed to be independent from each other.  We assume that each SU has access to its direct and interference channel gains by means of pilot training signals periodically transmitted by the SBS and PBS, {\em e.g.,} see \cite{Berry06, Boche07}.   
\begin{definition}\label{Def1}
We say that the cumulative distribution function (CDF) of a random variable $X$, denoted by $F_X\paren{x}$, belongs to the class $\cal C$-distributions if it satisfies the following properties:
\begin{itemize}
\item $F_X\paren{x}$ is continuous.
\item $F_X(x)$ has a positive support, {\em i.e.,} $F_X(x)=0$ for $x \leq 0$.
\item $F_X(x)$ is strictly increasing, {\em i.e.,} $F_X(x_1)<F_X(x_2)$ for $0<x_1<x_2$.
\item The tail of $F_X(x)$ decays to zero \emph{double exponentially}, {\em i.e.,} there exist constants $\alpha>0$, $\beta>0$, $n>0$, $l \in \R$ and a slowly varying function $H(x)$ satisfying $H(x) = \LO{x^{n}}$ as $x\ra\infty$ such that\footnote{By $p(x)=\LO{q(x)}$, we mean that $p(x)$ and $q(x)$ are two positive functions such that $\lim_{x \ra \infty}\frac{p(x)}{q(x)} = 0$.}   
  \begin{eqnarray}\label{Eq: tail-condition-1}
 \lim_{x\ra\infty}\frac{1-F_X(x)}{\alpha x^{l}\e{\paren{-\beta x^{n}+H(x)}}}=1.\nonumber
 \end{eqnarray}
\item $F_X(x)$ varies \emph{regularly} around the origin, {\em i.e.,} there exist constants $\eta>0$ and $\gamma>0$ such that
  \begin{eqnarray}\label{Eq: tail-condition-2}
 \lim_{x\downarrow0}\frac{F_X(x)}{\eta x^{\gamma}}=1.\nonumber
 \end{eqnarray}       
\end{itemize}
  \end{definition}

We assume that the CDFs of all fading power gains in this paper belong to the class $\cal C$-distributions. Table \ref{Table: FP} illustrates the parameters characterizing the behavior of the distribution of fading power gains around zero and infinity for the commonly used fading models in the literature. To avoid any confusion, these parameters are represented by subscript $h$ for direct channel gains and by subscript $g$ for interference channel gains in the rest of paper.

\begin{figure}[!t]
\centering{\includegraphics[scale=0.5]{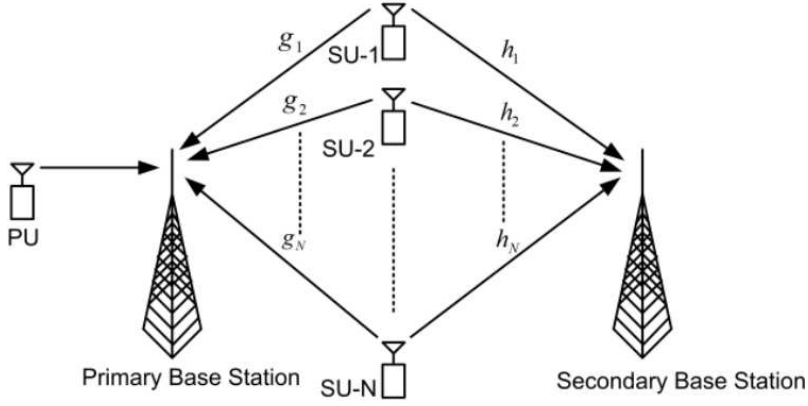}}
\caption{$N$ SUs forming a cognitive uplink to the SBS and interfering with signal reception at the PBS.} \label{F1}
\end{figure}

\begin{table}[!t]
\begin{minipage}{\textwidth}
\renewcommand{\arraystretch}{1.3}
\caption{Common fading channel models and their parameters}
\centering
\begin{tabular}{cccccccc}
\toprule
\multicolumn{1}{c}{\multirow{2}{*}{ Channel Model}}  & \multicolumn{7}{c}{ Parameters }\\
 \cmidrule(r){2-8}
&  $\alpha$ & $l$& $\beta$ & $n$ & $H(x)$ & $\eta$ & $\gamma$  \\
\midrule
Rayleigh & 1&  $0$ & 1 & 1 & 0 & 1&1 \\
\midrule
Rician-$K$& $\frac{1}{2\sqrt{\pi}\e{K}\sqrt[4]{K\paren{K+1}}}$ & $-\frac{1}{4}$ & $K+1$ &1 & $2\sqrt{K\paren{K+1}x}$& $\frac{K+1}{\e{K}}$ & 1\\
\midrule
Nakagami-$m$ & $\frac{m^{m-1}}{\Gamma(m)}$ & $m-1$ & $m$ & 1 & 0& $\frac{m^{m-1}}{\Gamma(m)}$& $m$ \\
\midrule
Weibull-$c$ & 1& 0 & $\Gamma^\frac{c}{2}\paren{1+\frac{2}{c}}$ & $\frac{c}{2}$ & $0$ & $\Gamma^\frac{c}{2}\paren{1+\frac{2}{c}}$ & $\frac{c}{2}$ \\\bottomrule
\label{Table: FP}
\end{tabular}
\end{minipage}
\end{table}

Each SU exploits knowledge of its direct and interference channel gains to \emph{locally} perform the task of power allocation, \emph{independent} of other SUs without any feedback from the SBS. A collision channel model is assumed for the resolution of concurrent transmissions from SUs at the SBS. That is, if more than one SUs transmit concurrently,  data transmissions from all of them collide, and the resulting throughput at the SBS is set to zero. In the next section, we derive the structure of the optimum distributed power control policy maximizing the secondary network sum-rate under the aforementioned assumptions for two different network types: \emph{(i)} distributed total power and interference limited (DTPIL) networks, and \emph{(ii)} distributed interference limited (DIL) networks.  After obtaining the optimum distributed power control policy, we also derive throughput scaling laws for these network types when each SU controls its transmission power optimally.  
\section{The Structure of the Optimum Distributed Power Control Policy and Throughput Scaling Laws}\label{Sec: R&D}
In this section, we will first present and solve the sum-rate maximization problems in DTPIL and DIL networks. Then, each problem will be followed by the corresponding throughput scaling results along with detailed insights into the observed throughput scaling behavior. All proofs are relegated to appendices for the sake of paper fluency. We start our discussions by formulating the sum-rate maximization problem for DTPIL networks.
\subsection{Optimum Power Control and Throughput Scaling in DTPIL Networks}
In DTPIL networks, transmission powers of SUs are limited by an average total transmission power constraint and a constraint on the average total interference power of SUs at the PBS. Transmission probabilities of SUs are also constrained to avoid excessive collisions. We define the power allocation policy in DTPIL networks, $P_{\rm DTPIL}\paren{\cdot,\cdot}$, as a mapping from $\R_+^2$ to $\R_+$, where $\PTPIL$ represents the transmission power of the $i$th SU at the joint channel state $\paren{h_i, g_i}$. The power allocation policy $P_{\rm DTPIL}\paren{\cdot,\cdot}$ is  designed such that the transmission probability is equal to $p_N$, $p_N \in (0, 1)$, for all SUs, \emph{i.e.,} $\PR{\PTPIL>0} = p_N$ for $i\in\left\{1,\cdots,N\right\}$.  Here, $p_N$ can be considered as a design degree-of-freedom helping us to keep the collision rate below some certain level. Under these modeling assumptions, the secondary network sum-rate for a given power control policy $P_{\rm DTPIL}$, $R_{\rm DTPIL}\paren{p_N,N,P_{\rm DTPIL}}$, can be expressed as 
\begin{eqnarray}\label{Eq: sum-rate}
 R_{\rm DTPIL}\paren{p_N,N,P_{\rm DTPIL}}&=&\ES{\sum_{i=1}^N\logp{1+h_i\PTPIL}\prod_{j\neq i}{\I{P_{\rm DTPIL}\paren{h_j,g_j}=0}}}\nonumber\\
&=&N\paren{1-p_N}^{N-1}\ES{\logp{1+hP_{\rm DTPIL}\paren{h,g}}},\nonumber
\end{eqnarray}
where $h$ and $g$ are two independent generic random variables distributed according to $F_h\paren{x}$ and $F_g\paren{x}$, respectively. Similarly,  the average total transmission power and the average interference power at the PBS can be written as 
\begin{eqnarray}\label{Eq: Power}
\ES{\sum_{i=1}^N\PTPIL}=N\ES{P_{\rm DTPIL}\paren{h,g}}\nonumber
\end{eqnarray}
and
\begin{eqnarray}\label{Eq: Interference}
\ES{\sum_{i=1}^Ng_i\PTPIL}=N\ES{gP_{\rm DTPIL}\paren{h,g}}\nonumber
\end{eqnarray}
respectively. In DTPIL networks, transmission powers of SUs are allocated according to the solution of the following functional optimization problem: 
\begin{eqnarray}\label{OP}
\begin{array}{ll}
\underset{P_{\rm DTPIL}\paren{h,g} \geq 0}{\mbox{maximize}} & R_{\rm DTPIL}\paren{p_N,N,P_{\rm DTPIL}} \\
\mbox{subject to} & N\EW_{h,g}\sqparen{P_{\rm DTPIL}\paren{h,g} } \leq P_{\rm ave} \\
			    & N\EW_{h,g} \sqparen{gP_{\rm DTPIL}\paren{h,g}} \leq Q_{\rm ave}\\
			    &\PR{P_{\rm DTPIL}\paren{h,g}>0}=p_N
\end{array}. 
\end{eqnarray}

The power optimization problem in \eqref{OP} is not necessarily a convex programming due to the transmission probability constraint. However, in the next theorem, we show that the optimum power control policy solving \eqref{OP} is in the form of a threshold based water-filling power control when the number of SUs is large enough. 
\begin{theorem}\label{Theo: OPA-TPIL}
Let $P^\star_{\rm DTPIL}\paren{h,g}$ be the solution of \eqref{OP}. Then, for $p_N=\frac{1}{N}$ and $N$ large enough, we have 
\begin{eqnarray}\label{Eq: OPA-TPIL}
P_{\rm DTPIL}^\star\paren{h,g}=\paren{\frac{1}{\lambda_N+\mu_N g}-\frac{1}{h}}^+\I{\frac{h}{\lambda_N+\mu_N g}>\ICRDTPIL{\frac{1}{N}}},
\end{eqnarray}
where $\lambda_N$ and $\mu_N$ are power control parameters adjusted such that the average total transmission power and the average interference power constraints in \eqref{OP} are met, and $F^{-1}_{\lambda_N,\mu_N}\paren{x}$ is the functional inverse of the CDF of $\CRTPIL$, \emph{i.e.,} $\CRDTPIL$.
\end{theorem}
\begin{proof}
See Appendix \ref{App: OPA-TPIL}.
\end{proof}

Theorem \ref{Theo: OPA-TPIL} pinpoints that the jointly optimal scheduling and power control strategy is in the form of a threshold-based water-filling power control policy. That is, the $i$th SU first decides against or in favor of transmission by comparing the value of its observed joint direct and interference channel state $\frac{h_i}{\lambda_N + \mu_N g_i}$ with the threshold value of $F^{-1}_{\lambda_N, \mu_N}\paren{1 - \frac{1}{N}}$. Upon a positive decision in favor of transmission, it transmits by using a water-filling power allocation policy, which is embodied by the $\paren{\frac{1}{\lambda_N + \mu_N g} - \frac{1}{h}}^+$ term in \eqref{Eq: OPA-TPIL}. Note that $\lambda_N$ and $\mu_N$ are computed off-line at the SBS by solving the dual problem associated with the optimization problem \eqref{OP-Auxiliary-2} in Appendix \ref{App: OPA-TPIL}. Then, the SBS broadcasts the values of $\lambda_N$ and $\mu_N$ to all SUs. 

In the centralized case, direct and interference channel gains of all SUs are available at the SBS, and in order to maximize the secondary network sum-rate, the SBS schedules the SU having the maximum of $\left\{\CRTPIL\right\}_{i=1}^N$ \cite{NID12}. The scheduled SU employs a water-filling power allocation policy with changing power levels. Hence, the multiuser diversity gain with a centralized scheduler depends on the maximum of $\left\{\CRTPIL\right\}_{i=1}^N$, which concentrates around $F^{-1}_{\lambda_N,\mu_N}\paren{1-\frac{1}{N}}$ as the number of SUs becomes large (\emph{i.e.,} see Lemma 2 in \cite{NIDSubmitted} for more details). Based on this observation and Theorem \ref{Theo: OPA-TPIL}, we conclude that in DTPIL networks, the $i$th SU transmits if the likelihood of its being the SU with the maximum of $\left\{\CRTPIL\right\}_{i=1}^N$ is high. Hence, throughput scaling laws similar to those obtained in \cite{NID12} and \cite{NIDSubmitted} are expected to hold for DTPIL networks when $p_N=\frac{1}{N}$. Later, we show that this choice of transmission probability is asymptotically optimal. That is, the secondary network sum-rate under $p_N=\frac{1}{N}$  serves as an upper bound on aggregate communication rates that we would otherwise achieve through other choices of $p_N$ when $N$ is large enough.
 
Under the collision channel assumption for resolving collisions, the SBS can decode the received signal successfully {\em if and only if} just one SU transmits. Otherwise, a collision happens and no data is delivered to the SBS. In our setup with $p_N = \frac{1}{N}$ and $N$ large enough, this observation implies that the received signal will be decoded successfully if and only if just the SU with the maximum of $\left\{\CRTPIL\right\}_{i=1}^N$ transmits. Let $\Xs{\lambda_N}{\mu_N}$ and $\Xd{\lambda_N}{\mu_N}$ be the largest and the second largest elements among the collection of i.i.d random variables $\brparen{X_i\paren{\lambda_N,\mu_N}}_{i=1}^N$, respectively, where $X_i\paren{\lambda_N,\mu_N}=\CRTPIL$. Let $\RTPIL{p_N}$ be the sum-rate in DTPIL networks under the optimum distributed power control policy with the transmission probability equal to $p_N$. Then, Theorem \ref{Theo: OPA-TPIL} implies that, for $p_N=\frac{1}{N}$ and $N$ large enough, we have 
\begin{eqnarray}\label{eq: sum-rate DTPIL}
\RTPIL{\frac{1}{N}}=\ES{\logp{\Xs{\lambda_N}{\mu_N}}\Inb{A_N}},
\end{eqnarray}
 where $A_N=\brparen{\Xs{\lambda_N}{\mu_N}>\ICRDTPIL{\frac{1}{N}},\Xd{\lambda_N}{\mu_N}\leq \ICRDTPIL{\frac{1}{N}}}$.  
 In the next theorem, we derive the scaling behavior of $\RTPIL{\frac{1}{N}}$.
\begin{theorem}\label{Theo: DTPIL}
The secondary network sum-rate $\RTPIL{\frac{1}{N}}$ for $p_N=\frac{1}{N}$ under the optimum distributed power control policy scales according to  
\begin{eqnarray}
\lim_{N\ra\infty}\frac{\RTPIL{\frac{1}{N}}}{\log\logp{N}}=\frac{1}{\e{} n_h}.\nonumber
\end{eqnarray}
\end{theorem}
\begin{IEEEproof}
See Appendix \ref{app: TPIL}.
\end{IEEEproof}

Theorem \ref{Theo: DTPIL} formally establishes the double logarithmic scaling behavior of the secondary network sum-rate for DTPIL networks. Further, it shows that the pre-log multiplier in this scaling behavior is equal to $\frac{1}{\e{} n_h}$. $n_h$ is equal to $1$ for Rayleigh, Rician-$K$ and Nakagami-$m$ distributed direct channel gains, and is equal to $\frac{c}{2}$ for Weibull-$c$ distributed direct channel gains.
 
The result of Theorem \ref{Theo: DTPIL} has the following intuitive explanation. The event $A_N$ in \eqref{eq: sum-rate DTPIL} represents the successful transmission event. For $N$ large enough, $\PRP{A_N}$ represents the fraction of time that only the SU with the maximum of $\left\{\CRTPIL\right\}_{i=1}^N$ transmits. In Appendix \ref{app: TPIL}, we show that, for $p_N=\frac{1}{N}$, $\PRP{A_N}$ converges to $\frac{1}{\e{}}$ as $N$ becomes large. Hence, as the number of SUs becomes large, the fraction of time that just the best SU transmits is approximately equal to $\frac{1}{\e{}}$. Also, in Appendix \ref{app: TPIL}, we show that $\logp{\Xs{\lambda_N}{\mu_N}}$ term in \eqref{eq: sum-rate DTPIL} scales according to $\frac{1}{n_h}\log\logp{N}$. These observations suggest that the secondary network sum-rate (under the optimum distributed power control) scales according to $\frac{1}{\e{}n_h}\log\logp{N}$ as $N$ becomes large. It would be noted that this is just an intuitive explanation, and we provide the rigorous proof in Appendix \ref{app: TPIL}. 
 
We also identify the second order determinants of $\RTPIL{\frac{1}{N}}$ in Appendix \ref{app: TPIL}.  Formally, we show that it can be expressed as 
 \begin{eqnarray}\label{Eq: Rate-TPIL}
\RTPIL{\frac{1}{N}}=\logp{\frac{1}{\lambda_N}}\PRP{A_N}+\ES{\logp{\Xs{1}{\frac{\mu_N}{\lambda_N}}}\Inb{A_N}}.
\end{eqnarray}
We show that the first term in \eqref{Eq: Rate-TPIL} converges to $\frac{1}{\e{}}\logp{P_{\rm ave}}$ as $N$ becomes large. This finding displays the logarithmic effect of the power constraint on the secondary sum-rate in DTPIL networks.  The second term in \eqref{Eq: Rate-TPIL} gives rise to the scaling of secondary sum-rate according to $\frac{1}{\e{} n_h}\log\log(N)$. 
 
 So far, we have assumed that $p_N$ is equal to $\frac{1}{N}$. One may speculate that DTPIL networks may obtain a better throughput scaling behavior if the transmission probability is optimally adjusted, rather than to be set to $\frac{1}{N}$. To investigate this idea, we study the throughput scaling behavior of DTPIL networks under the optimum transmission probability selection in the next theorem.
 
\begin{theorem}\label{Theo: Opt-Scal-DTPIL}
For each $N\in \N$, let $p^\star_N$ be an optimum transmission probability selection maximizing $\RTPIL{p_N}$, \emph{i.e.,} $p^\star_N \in \arg\max_{0\leq p_N\leq 1}\RTPIL{p_N}$. Then, 
\begin{eqnarray}
\lim_{N\ra\infty}\frac{\RTPIL{p^\star_N}}{\log\logp{N}}= \frac{1}{\e{}n_h}.\nonumber
\end{eqnarray}
\end{theorem}
\begin{IEEEproof}
See Appendix \ref{App: Opt-Scal-DTPIL}.
\end{IEEEproof}

Theorem \ref{Theo: Opt-Scal-DTPIL} indicates that the secondary network sum-rate under the optimum transmission probability also scales according to $\frac{1}{\e{}n_h}\log\logp{N}$. Thus, the choice of transmission probability as $p_N=\frac{1}{N}$ is asymptotically optimal, and the secondary network achieves the same throughput scaling under $p^\star_N$ and $p_N=\frac{1}{N}$. However, it should be noted that the optimum transmission probability might be different from $\frac{1}{N}$ for any finite $N$.  Identical throughput scaling behavior of DTPIL networks under $p^\star_N$ and $p_N=\frac{1}{N}$ gives rise to the following question: Does the optimum transmission probability asymptotically behaves as $\frac{1}{N}$? The next lemma gives an affirmative answer to this question.
\begin{lemma}\label{Theo: Opt-Pro-DTPIL}
For each $N \in \N$, let $p^\star_N$ be an optimum transmission probability selection in DTPIL networks. Then, $\lim_{N\ra\infty}Np^\star_N=1$.
\end{lemma}\begin{IEEEproof}
See Appendix \ref{App: Opt-Pro-DIL}.
\end{IEEEproof}

Lemma \ref{Theo: Opt-Pro-DTPIL} shows that the optimum transmission probability in DTPIL networks should scale according to $\frac{1}{N}$.  This scaling behavior of $p^\star_N$ can be intuitively considered as the origin of identical throughput scaling behavior of DTPIL networks under $p^\star_N$ and $p_N=\frac{1}{N}$ (\emph{i.e.,} see Appendix \ref{App: Opt-Pro-DTPIL} for more details). 

Finally, it is perceptive to compare the throughput scaling laws obtained by using completely decentralized transmission strategies with those obtained through a centralized scheduler. 
In \cite{NIDSubmitted}, it has been shown that the secondary network throughput  with a centralized scheduler (usually, the SBS) scales according to $\frac{1}{n_h}\log\logp{N}$ when the optimum power allocation policy is employed. Hence, compared to the centralized case, the factor $\frac{1}{\e{}}$ here can be interpreted as the price of avoiding feedback signals between primary and secondary networks, which are  the key parameters required by the centralized scheduler to perform optimum power control and scheduling. 
\subsection{Optimum Power Control and Throughput Scaling in DIL Networks}
In this case, transmission powers of SUs are limited by a constraint on the total average interference power that SUs cause to the PBS and a transmission probability constraint. We define the power allocation policy in DIL networks, $P_{\rm DIL}\paren{\cdot,\cdot}$, as a mapping from $\R_+^2$ to $\R_+$, where $\PIL$ denotes the transmission power of the $i$th SU at the joint channel state $\paren{h_i,g_i}$. Similar to DTPIL networks, the power allocation policy in DIL networks is designed such that the transmission probability for all SUs is equal to $p_N$, \emph{i.e.,} $\PR{\PIL>0}=p_N$ for all $i\in\left\{1,\cdots,N\right\}$.  We define $R_{\rm DIL}\paren{p_N,N,P_{\rm DIL}}$ as the secondary network sum-rate for a given power control policy $P_{\rm DIL}$ in DIL networks, which can be expressed as $R_{\rm DIL}\paren{p_N,N,P_{\rm DIL}} = N\paren{1-p_N}^{N-1}\ES{\logp{1+hP_{\rm DIL}\paren{h,g}}}$.  In this case, transmission powers of SUs are allocated according to the solution of the following functional optimization problem:
\begin{eqnarray}\label{OP-IL}
\begin{array}{ll}
\underset{P_{\rm DIL}\paren{h,g} \geq 0}{\mbox{maximize}} & R_{\rm DIL}\paren{p_N,N,P_{\rm DIL}} \\
\mbox{subject to} & N\EW_{h,g} \sqparen{gP_{\rm DIL}\paren{h,g}} \leq Q_{\rm ave}\\
			    &\PR{P_{\rm DIL}\paren{h,g}>0}=p_N
\end{array}. 
\end{eqnarray}
The next theorem establishes the structure of the optimum power allocation policy in DIL networks. The proof of Theorem \ref{Theo: OPA-IL} is similar to that of Theorem \ref{Theo: OPA-TPIL}, and therefore, it is skipped to avoid repetition.
\begin{theorem}\label{Theo: OPA-IL}
Let $P^\star_{\rm DIL}\paren{h,g}$ be the solution of \eqref{OP-IL}. Then, for  $p_N=\frac{1}{N}$ and $N$ large enough, we have 
\begin{eqnarray}
P_{\rm DIL}^\star\paren{h,g}=\paren{\frac{1}{\mu_N g}-\frac{1}{h}}^+\I{\frac{h}{g}>\ICRDIL{\frac{1}{N}}},\nonumber
\end{eqnarray}
where $\mu_N$ is the power control parameter adjusted such that the average interference power constraint in \eqref{OP-IL} is met with equality, and $F^{-1}_{\frac{h}{g}}\paren{x}$ is the functional inverse of the CDF of $\CRIL$.
\end{theorem}

Theorem \ref{Theo: OPA-IL} implies that, for $p_N=\frac{1}{N}$ and $N$ large enough, the optimum power allocation policy for the $i$th SU is to transmit by using a water-filling power allocation policy if its joint power and interference channel state, \emph{i.e.,} $\CRIL$, is greater than the threshold value of $\ICRDIL{\frac{1}{N}}$. In the centralized case, in order to maximize the secondary network sum-rate, the SBS schedules the SU having the maximum of $\left\{\CRIL\right\}_{i=1}^N$, and the scheduled SU employs a water-filling power allocation policy, \emph{i.e.,} see \cite{NID12} and \cite{NIDSubmitted}. Moreover, the multiuser diversity gain with a centralized scheduler heavily depends on the maximum of $\left\{\frac{h_i}{g_i}\right\}_{i=1}^N$, and as the number of SUs becomes large,  the maximum of $\left\{\CRIL\right\}_{i=1}^N$ takes values around $\ICRDIL{\frac{1}{N}}$ with high probability.  In this regard, Theorem \ref{Theo: OPA-IL} further shows that, in DIL networks, a SU transmits with positive power if it has a high chance of being the SU with the maximum of $\left\{\CRIL\right\}_{i=1}^N$. Thus, we expect to observe throughput scaling behavior similar to that observed with a centralized scheduler, which is indeed the case as shown next. 
 
Let $\Ys$ and $\Yd$ be the largest and the second largest elements of the collection of random variables $\brparen{Y_i}_{i=1}^N$, where $Y_i=\CRIL$. Also, let $\RIL{p_N}$ be the sum-rate in DIL networks under the optimum distributed power control policy with transmission probability equal to $p_N$. Then, for $p_N=\frac{1}{N}$ and $N$ large enough, we have 
\begin{eqnarray}
\RIL{\frac{1}{N}}=\ES{\logp{\frac{\Ys}{\mu_N}}\Inb{B_N}},\nonumber
\end{eqnarray}
where $B_N=\brparen{\Ys>\ICRDIL{\frac{1}{N}},\Yd\leq \ICRDIL{\frac{1}{N}}}$.
The next theorem establishes the sum-rate scaling behavior of DIL networks.
\begin{theorem}\label{Theo: DIL}
The secondary network sum-rate $\RIL{\frac{1}{N}}$ for $p_N=\frac{1}{N}$ under the optimum distributed power control policy scales according to
\begin{eqnarray}
\lim_{N\ra\infty}\frac{\RIL{\frac{1}{N}}}{\logp{N}}=\frac{1}{\e{}\gamma_g}.\nonumber
\end{eqnarray}
\end{theorem}
\begin{IEEEproof}
See Appendix \ref{app: IL}.
\end{IEEEproof}

Theorem \ref{Theo: DIL} reveals the logarithmic scaling behavior of the secondary network sum-rate as a function of  the number of SUs in DIL networks. It also shows that the pre-log multiplier in this scaling behavior is equal to $\frac{1}{\e{}\gamma_g}$.  The $\frac{1}{\e{}}$ stems from the probability of successful transmission, whereas $\frac{1}{\gamma_g}$ term stems from the throughput scaling behavior on the event of successful transmission.  $\gamma_g$ is equal to 1 for Rayleigh and Rician-$K$ distributed interference channel gains, and is equal to $m$ and $\frac{c}{2}$ when interference channel gains are Nakagami-$m$ and Weibull-$c$ distributed, respectively.

Our analysis in Appendix \ref{app: IL} also reveals some second order effects on the secondary network throughput. In particular, we show that $\RIL{\frac{1}{N}}$ can be written as 
\begin{eqnarray}\label{Eq: Rate-Expan-DIL}
\RIL{\frac{1}{N}}=\logp{\frac{1}{\mu_N}}\PRP{B_N}+\ES{\logp{\Ys}\Inb{B_N}}.
\end{eqnarray}
It is shown in Appendix \ref{app: IL} that $\mu_N$ converges to $\frac{1}{Q_{\rm ave}}$ as $N$ grows large. Hence, the first term in \eqref{Eq: Rate-Expan-DIL} converges to $\frac{1}{\e{}}\logp{Q_{\rm ave}}$ as $N$ becomes large, implying the logarithmic effect of $Q_{\rm ave}$ on the secondary network sum-rate in DIL networks.  Further, it is also shown that the second term in \eqref{Eq: Rate-Expan-DIL} scales according to $\frac{1}{\e{}\gamma_g}\logp{N}$, signifying the logarithmic effect of the number of SUs on the secondary network sum-rate in DIL networks. 

It is also instructive to compare the result of Theorem \ref{Theo: DIL} with the throughput scaling behavior that can be obtained by means of a centralized scheduler. The secondary network throughput scales according to $\frac{1}{\gamma_g}\logp{N}$ when the optimum transmission power control is  performed by a centralized scheduler \cite{NIDSubmitted}. This observation suggests that  the throughput scaling law obtained through distributed implementation differs from that obtained in the centralized case only in the observed pre-log factors. Similar to the previous case, $\frac{1}{\e{}}$ can be interpreted as the cost of decentralized implementation of the cognitive uplink.

It might be hypothesized that the capacity scaling behavior obtained in Theorem \ref{Theo: DIL} can be improved if the optimum transmission probability is employed instead of $\frac{1}{N}$. The next theorem disproves this hypothesis.
\begin{theorem}\label{Theo: Opt-Scal-DIL}
For each $N\in \N$, let $p^\star_N$ be an optimum transmission probability selection maximizing $\RIL{p_N}$, \emph{i.e.,} $p^\star_N \in \arg\max_{0\leq p_N\leq 1}\RIL{p_N}$. Then, 
\begin{eqnarray}
\lim_{N\ra\infty}\frac{\RIL{p^\star_N}}{\logp{N}}= \frac{1}{\e{}\gamma_g}.
\end{eqnarray}
\end{theorem}
\begin{IEEEproof}
See Appendix \ref{App: Opt-Scal-DIL}. 
\end{IEEEproof}  

Theorem \ref{Theo: Opt-Scal-DIL} establishes the logarithmic throughput scaling behavior of DIL networks under the optimum transmission probability. Hence, the choice of $p_N=\frac{1}{N}$ is asymptotically optimal, and one cannot obtain better throughput scaling by other choices of $p_N$. Finally, in the next lemma, we study the asymptotic behavior of the sequence of optimum transmission probabilities in DIL networks as the number of SUs becomes large.
\begin{lemma}\label{Theo: Opt-Pro-DIL}
For each $N\in \N$, let $p^\star_N$ be an optimum transmission probability selection in DIL networks. Then, $\lim_{N\ra\infty}Np^\star_N=1$.
\end{lemma}
\begin{IEEEproof}
See Appendix \ref{App: Opt-Pro-DIL}.
\end{IEEEproof} 

Lemma \ref{Theo: Opt-Pro-DIL} indicates that the sequence of optimum transmission probabilities in DIL networks decays to zero at the rate of $\frac{1}{N}$. In other words, $\frac{1}{N}$ serves as a good approximation for the optimum transmission probability when $N$ is large enough.

\section{Numerical Results}\label{Sec: NR}
In this section, we numerically evaluate the sum-rate performance of DTPIL and DIL networks as a function of the number of SUs. We also compare their sum-rate scaling behavior with that of orthogonal channel access networks, \emph{i.e.,} time division multiple access (TDMA) and frequency division multiple access (FDMA) networks. In the considered orthogonal channel access schemes, the global CSI is not available at the SBS either, and communication resources ({\em i.e.,} either time or frequency) are periodically allotted to SUs regardless of their channel conditions. 

In TDMA networks, time is divided into equal length time slots, and a time slot is allocated to each SU. In FDMA networks, the total frequency band is divided into narrow-band frequency chunks, and each frequency chuck is allocated to a SU.  In orthogonal channel access networks, we assume that SUs have access to their direct and interference channel gains, and upon being scheduled for transmission, each SU adjusts its transmission power level according to a single-user water-filling power allocation policy based on the local knowledge of its channel gains. The same average total transmission power and average interference power constraints are considered for DTPIL, DIL and orthogonal channel access networks.
\begin{figure}[t]
\centering
\subfigure[]
{
\begin{tikzpicture}
\begin{axis}[title={Throughput in DTPIL and orthogonal access networks},legend style={yshift=-3cm, xshift=0.05cm},\axissetting]
                                      
\pfgsetting
\addplot+[color=blue,mark=none] table [x=N,y=43loglog]{WeiRTPIL.dat};\addlegendentry{$\frac{2}{1.5\e{}}\log\logp{N}+\frac{1}{\e{}}\logp{P_{\rm ave}}$};
\addplot+[color=magenta,mark=none,dashed] table [x=N,y=c15]{WeiRTPIL.dat};\addlegendentry{DTPIL (Weibull, $c=1.5$)};
\addplot+[color=black,mark=none,dotted] table [x=N,y=c1.5TDMA]{WeiRTPIL.dat};\addlegendentry{Orthogonal Channel Access};
\end{axis}
\end{tikzpicture}
\label{WeiRTPIL1}}
\subfigure[]
{
\begin{tikzpicture}
\begin{axis}[title={Throughput in DTPIL and orthogonal access networks},legend style={yshift=-3cm, xshift=0.05cm},\axissetting]
                                      
\pfgsetting
\addplot+[color=blue,mark=none] table [x=N,y=2025loglog]{WeiRTPIL.dat};\addlegendentry{$\frac{2}{2.5\e{}}\log\logp{N}+\frac{1}{\e{}}\logp{P_{\rm ave}}$};
\addplot+[color=magenta,mark=none,dashed] table [x=N,y=c25]{WeiRTPIL.dat};\addlegendentry{DTPIL (Weibull, $c=2.5$)};
\addplot+[color=black,mark=none,dotted] table [x=N,y=c2.5TDMA]{WeiRTPIL.dat};\addlegendentry{Orthogonal Channel Access};
\end{axis}
\end{tikzpicture}
\label{WeiRTPIL2}}
\subfigure[]
{
\begin{tikzpicture}
\begin{axis}[title={Throughput in DTPIL and orthogonal access networks},legend style={yshift=-3cm, xshift=0.05cm},\axissetting]
                                      
\pfgsetting
\addplot+[color=blue,mark=none] table [x=N,y=loglog]{RWeiTPIL.dat};\addlegendentry{$\frac{1}{\e{}}\log\logp{N}+\frac{1}{\e{}}\logp{P_{\rm ave}}$};
\addplot+[color=magenta,mark=none,dashed] table [x=N,y=c1.5]{RWeiTPIL.dat};\addlegendentry{DTPIL (Rayleigh)};
\addplot+[color=black,mark=none,dotted] table [x=N,y=c1.5TDMA]{RWeiTPIL.dat};\addlegendentry{Orthogonal Channel Access};

\end{axis}
\end{tikzpicture}
\label{RWeiTPIL}}
\subfigure[]
{
\begin{tikzpicture}
\begin{axis}[title={Throughput in DTPIL networks},legend style={yshift=-3cm, xshift=0.05cm},\axissetting]
                                      
\pfgsetting
\addplot+[color=blue,mark=none] table [x=N,y=1/N]{RRPTPIL.dat};\addlegendentry{$p_N=\frac{1}{N}$};
\addplot+[color=magenta,mark=none,dashed] table [x=N,y=1/4N]{RRPTPIL.dat};\addlegendentry{$p_N=\frac{1}{4N}$};
\addplot+[color=black,mark=none,dotted] table [x=N,y=1/10N]{RRPTPIL.dat};\addlegendentry{$p_N=\frac{1}{10N}$};
\end{axis}
\end{tikzpicture}
\label{RRPTPIL}}
\caption{ Secondary network throughput in DTPIL and orthogonal channel access networks as a function of the number of SUs for different communication environments (a)-(c). Throughput in DTPIL networks as a function of the number of SUs for different choices of $p_N$ (d). $P_{\rm ave}$ and $Q_{\rm ave}$ are set to 15dB and 0dB, respectively.}
\label{FTPIL}
\end{figure}
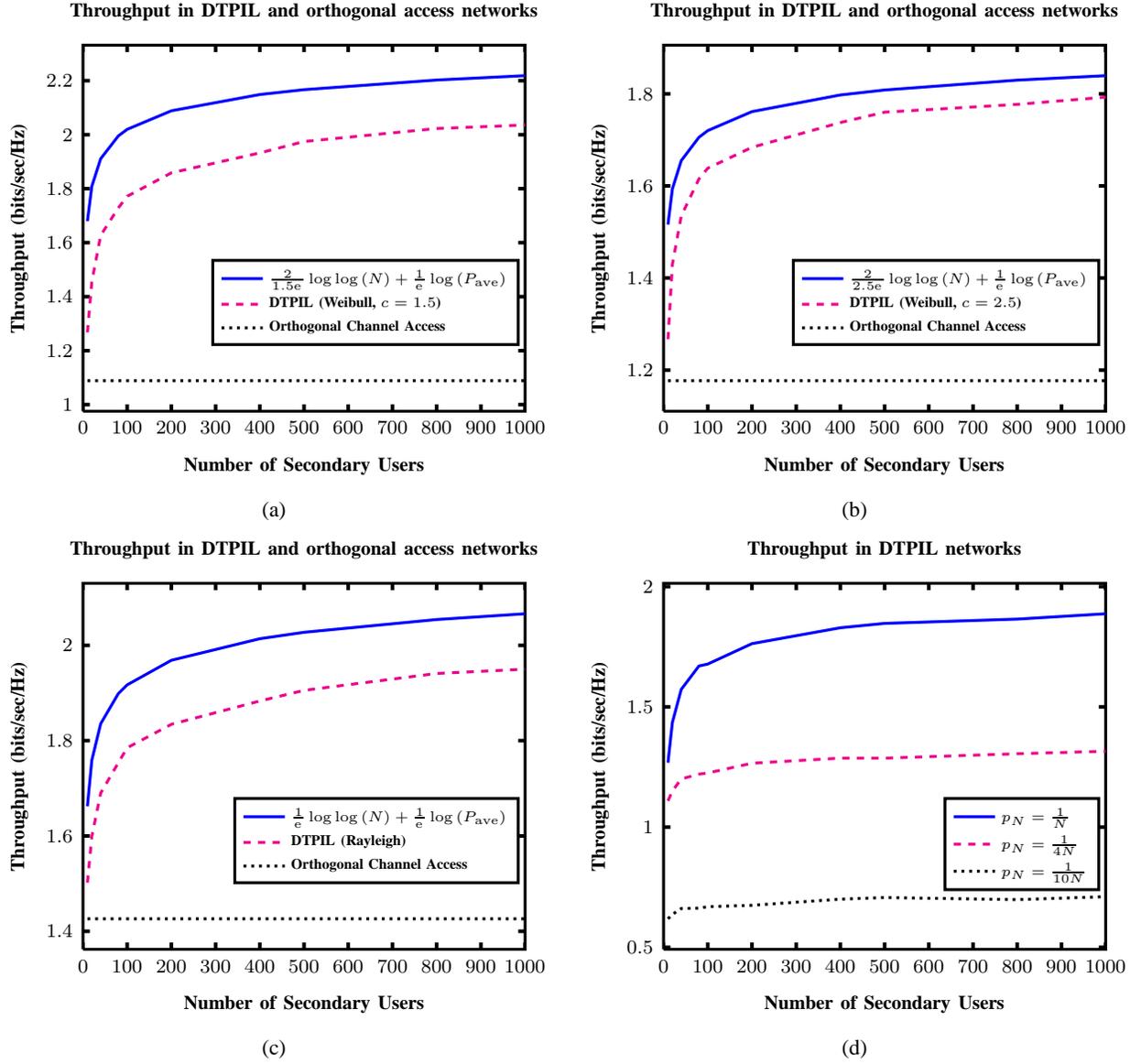

Figures \ref{FTPIL}(a)-(c) demonstrate the sum-rate scaling behavior of DTPIL and orthogonal channel access networks as a function of the number of SUs. In these figures, $P_{\rm ave}$ and $Q_{\rm ave}$ are set to 15dB and 0dB, respectively. Similar qualitative behavior continues to hold for other values of $P_{\rm ave}$ and $Q_{\rm ave}$. In Figs. \ref{FTPIL}(a)-(c), $p_N$ is set to $\frac{1}{N}$ for DTPIL networks. Also, identical fading models are considered for both DTPIL and orthogonal channel access networks.

 More specifically, in Fig. \ref{WeiRTPIL1}, direct channel gains are distributed according to the Weibull-$c$ fading model with $c=1.5$ and interference channel gains are distributed according to the Rayleigh fading model. In Fig. \ref{WeiRTPIL2}, direct channel gains are Weibull-$c$ distributed with $c=2.5$ and interference channel gains are Rayleigh distributed. As Fig. \ref{WeiRTPIL1} and Fig. \ref{WeiRTPIL2} show, the secondary network sum-rate in DTPIL networks scales according to $\frac{2}{\e{}c}\log\logp{N}$ with the number of SUs, \emph{i.e.,} $\frac{2}{1.5\e{}}\log\logp{N}$ for $c=1.5$ and $\frac{2}{2.5\e{}}\log\logp{N}$ for $c=2.5$, when direct channel gains are Weibull-$c$ distributed, a behavior predicted by Theorem \ref{Theo: DTPIL}. Also,  closeness of the simulated data rates of DTPIL networks to the curves of $\frac{2}{1.5\e{}}\log\logp{N}+\frac{1}{\e{}}\logp{P_{\rm ave}}$ and $\frac{2}{2.5\e{}}\log\logp{N}+\frac{1}{\e{}}\logp{P_{\rm ave}}$ in Fig. \ref{WeiRTPIL1} and Fig. \ref{WeiRTPIL2}, respectively, indicates the logarithmic effect of $P_{\rm ave}$ on the secondary network throughput in DTPIL networks.

In Fig. \ref{RWeiTPIL}, direct channel gains are Rayleigh distributed and interference channel gains are Weibull-$c$ distributed with $c=1.5$. As Fig. \ref{RWeiTPIL} shows, the secondary network sum-rate in DTPIL networks scales according to $\frac{1}{\e{}}\log\logp{N}$ when direct channel gains are Rayleigh distributed, which is also in accordance with Theorem \ref{Theo: DTPIL}. Also, proximity of simulated data rates of DTPIL network to the $\frac{1}{\e{}}\log\logp{N}+\frac{1}{\e{}}\logp{P_{\rm ave}}$ curve in this figure again shows the logarithmic effect of $P_{\rm ave}$ on the secondary network sum-rate in DTPIL networks. 

Moreover, as Figs. \ref{FTPIL}(a)-(c) show, the secondary network throughput in orthogonal channel access networks does not scale with the number of SUs since SUs are scheduled for transmission regardless of their channel gains, rather than being scheduled opportunistically, in these networks.  Furthermore, DTPIL networks achieve higher throughputs compared to those achieved by orthogonal channel access networks, even with possibly suboptimal choice of transmission probability (\emph{i.e.,} $p_N=\frac{1}{N}$) for small numbers of SUs.  This is due to the fact that DTPIL networks can harvest multiuser diversity gains, in a distributed fashion, without any global knowledge of CSI at the SBS. 

In Fig. \ref{RRPTPIL}, we demonstrate the secondary network throughput scaling in DTPIL networks as a function of the number of SUs when $p_N$ is set to $\frac{1}{N},\frac{1}{4N}$ and $\frac{1}{10N}$. In this figure, direct and interference channel gains are distributed according to the Rayleigh fading model. As Fig. \ref{RRPTPIL} shows, the secondary network asymptotically achieves much higher throughputs with $p_N=\frac{1}{N}$ when compared to other choices of $p_N$ that do not scale according to $\frac{1}{N}$. This finding signifies the importance of setting $p_N$ correctly to maximize secondary network sum-rates in DTPIL networks. 

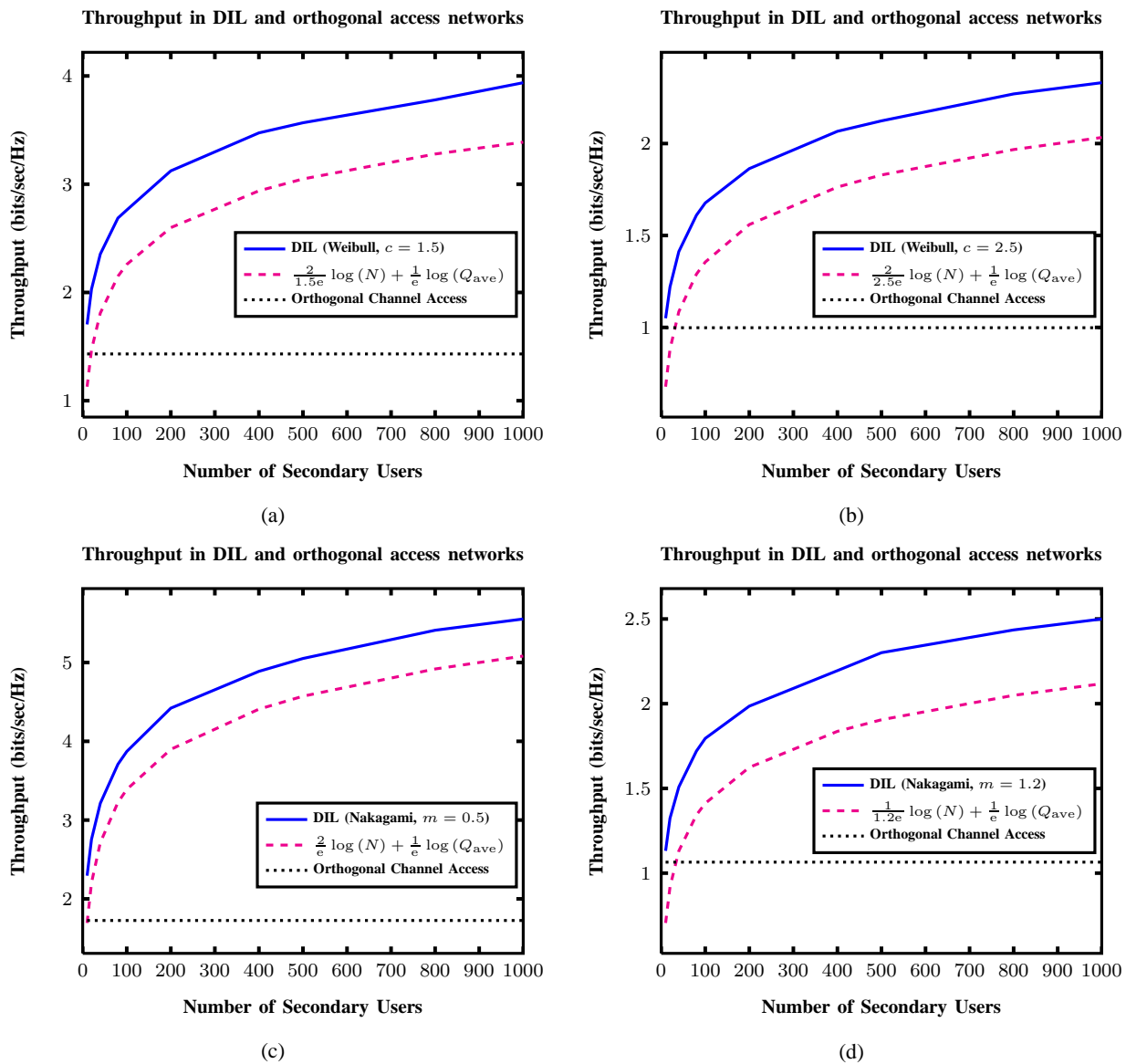
\begin{figure}[t]
\centering
\subfigure[]
{
\begin{tikzpicture}
\begin{axis}[title={Throughput in DIL and orthogonal access networks},legend style={yshift=-2.5cm, xshift=0.05cm},\axissetting]
                                      
\pfgsetting
\addplot+[color=blue,mark=none] table [x=N,y=c1.5]{RWeiIL.dat};\addlegendentry{DIL (Weibull, $c=1.5$)};
\addplot+[color=magenta,mark=none,dashed] table [x=N,y=4/3log]{RWeiIL.dat};\addlegendentry{$\frac{2}{1.5\e{}}\logp{N}+\frac{1}{\e{}}\logp{Q_{\rm ave}}$};
\addplot+[color=black,mark=none,dotted] table [x=N,y=c1.5TDMA]{RWeiIL.dat};\addlegendentry{Orthogonal Channel Access};
\end{axis}
\end{tikzpicture}
\label{RWeiIL1}}
\subfigure[]
{
\begin{tikzpicture}
\begin{axis}[title={Throughput in DIL and orthogonal access networks},legend style={yshift=-2.5cm, xshift=0.05cm},\axissetting]
                                      
\pfgsetting
\addplot+[color=blue,mark=none] table [x=N,y=c2.5]{RWeiIL.dat};\addlegendentry{DIL (Weibull, $c=2.5$)};
\addplot+[color=magenta,mark=none,dashed] table [x=N,y=4/5log]{RWeiIL.dat};\addlegendentry{$\frac{2}{2.5\e{}}\logp{N}+\frac{1}{\e{}}\logp{Q_{\rm ave}}$};
\addplot+[color=black,mark=none, dotted] table [x=N,y=c2.5TDMA]{RWeiIL.dat};\addlegendentry{Orthogonal Channel Access};
\end{axis}
\end{tikzpicture}
\label{RWeiIL2}}
\subfigure[]
{
\begin{tikzpicture}
\begin{axis}[title={Throughput in DIL and orthogonal access networks},legend style={yshift=-3cm, xshift=0.05cm},\axissetting]
                                  
\pfgsetting
\addplot+[color=blue,mark=none] table [x=N,y=m0.5]{RNakIL.dat};\addlegendentry{DIL (Nakagami, $m=0.5$)};
\addplot+[color=magenta,mark=none,dashed] table [x=N,y=2log]{RNakIL.dat};\addlegendentry{$\frac{2}{\e{}}\logp{N}+\frac{1}{\e{}}\logp{Q_{\rm ave}}$};
\addplot+[color=black,mark=none,dotted] table [x=N,y=m0.5TDMA]{RNakIL.dat};\addlegendentry{Orthogonal Channel Access};
\end{axis}
\end{tikzpicture}
\label{RNaIL1}}
\subfigure[]
{
\begin{tikzpicture}
\begin{axis}[title={Throughput in DIL and orthogonal access networks},legend style={yshift=-2.5cm, xshift=0.05cm},\axissetting]
                                      
\pfgsetting
\addplot+[color=blue,mark=none] table [x=N,y=m1.2]{RNakIL.dat};\addlegendentry{DIL (Nakagami, $m=1.2$)};
\addplot+[color=magenta,mark=none,dashed] table [x=N,y=0.8log]{RNakIL.dat};\addlegendentry{$\frac{1}{1.2\e{}}\logp{N}+\frac{1}{\e{}}\logp{Q_{\rm ave}}$};
\addplot+[color=black,mark=none, dotted] table [x=N,y=m1.2TDMA]{RNakIL.dat};\addlegendentry{Orthogonal Channel Access};

\end{axis}
\end{tikzpicture}
\label{RNaIL2}}
\caption{Secondary network throughput in DIL and orthogonal channel access networks as a function of the number of SUs for different communication environments (a)-(d). $Q_{\rm ave}$ is set to 0dB.}
\label{FIL}
\end{figure}

Figure \ref{FIL} shows the change of the secondary network sum-rate in DIL and orthogonal channel access networks as a function of the number of SUs for different communication environments. In this figure, $Q_{\rm ave}$ is set to 0dB. Similar qualitative behavior continues to hold for other values of $Q_{\rm ave}$. The transmission probability is set to $\frac{1}{N}$ for DIL networks. In Fig. \ref{RWeiIL1}, direct channel gains are distributed according to the Rayleigh fading model and interference channel gains are distributed according to the Weibull-$c$ fading model with $c=1.5$.  In Fig. \ref{RWeiIL2}, direct channel gains are Rayleigh distributed and interference channel gains are Weibull-$c$ distributed with $c=2.5$. From these figures, we can clearly observe that the secondary network throughput scales according to $\frac{2}{\e{} c} \log(N)$ as a function of the number of SUs when interference channel gains are Weibull distributed with different values of $c$. 

In Figs. \ref{RNaIL1} and \ref{RNaIL2}, direct channel gains are Rayleigh distributed and interference channel gains are Nakagami-$m$ distributed with $m$ set to $0.5$ and $1.2$, respectively. As these figures indicate, the secondary network throughput scales according to $\frac{1}{\e{}m}\logp{N}$ with the number of SUs in DIL networks  for Nakagami-$m$ distributed interference channel gains. All simulated capacity curves in Fig. 3 concur with the capacity scaling laws established in Theorem \ref{Theo: DIL}. We know that $p_N = \frac{1}{N}$ may not be the optimum choice of transmission probability for $N$ small enough, but from Fig. \ref{FIL}, we still observe that DIL networks with  $p_N = \frac{1}{N}$ outperform orthogonal channel access networks largely, in terms of the sum-rate performance, even for small numbers of SUs.

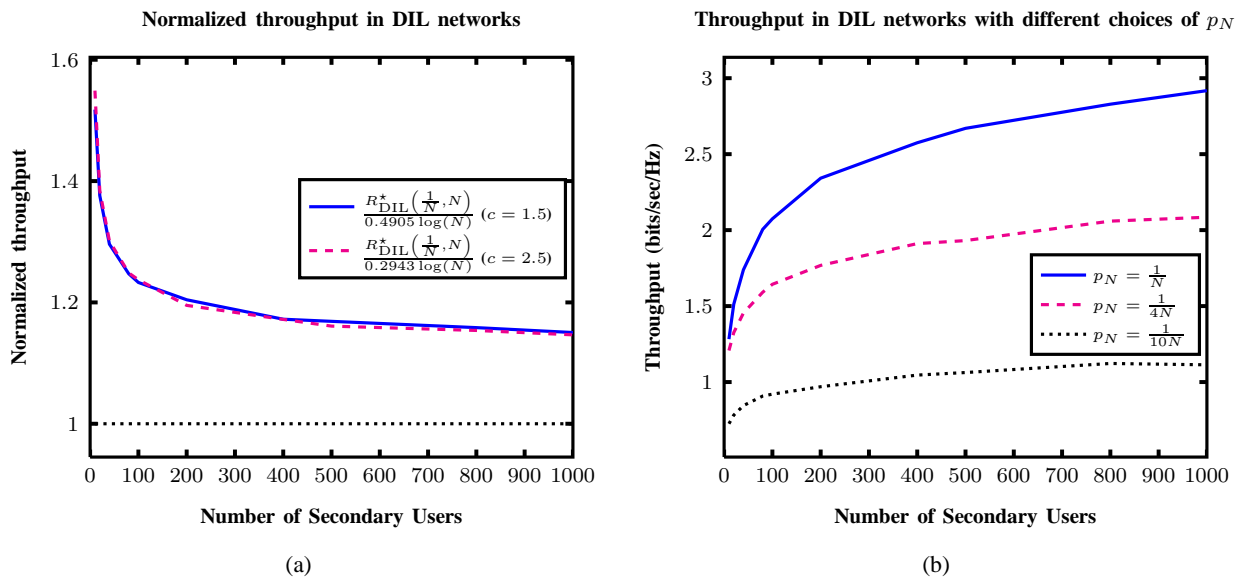
\begin{figure}
\centering
\subfigure[]
{
\begin{tikzpicture}
\begin{axis}[title={Normalized throughput in DIL networks},legend style={yshift=-1.5cm, xshift=0.05cm},\axissetting]
                                      
\pgfplotsset{
width= 8cm, 
every axis/.append style={line width=1.2pt},
label style={font=\bf\scriptsize}, 
ylabel style={yshift=-0.8em},
xlabel={Number of Secondary Users},
ylabel={Normalized throughput},
title style={font=\bf\scriptsize}, 
tick label style={font=\scriptsize,/pgf/number format/1000 sep={} },
tick style={ line width=1.5pt},
legend style={font=\bf\tiny,cells={anchor=west},solid},
every mark/.append style={solid},
}
\addplot+[color=blue,mark=none] table [x=N,y=c1.5]{RWeiILR.dat};\addlegendentry{$\frac{R^\star_{\rm DIL}\paren{\frac{1}{N},N}}{0.4905\logp{N}}$ ($c=1.5$)};
\addplot+[color=magenta,mark=none,dashed] table [x=N,y=c2.5]{RWeiILR.dat};\addlegendentry{$\frac{R^\star_{\rm DIL}\paren{\frac{1}{N},N}}{ 0.2943\logp{N}}$ ($c=2.5$)};
\addplot+[color=black,mark=none,dotted] table [x=N,y=1]{RWeiILR.dat};
\end{axis}
\end{tikzpicture}
\label{RWeiILR}}
\subfigure[]
{
\begin{tikzpicture}
\begin{axis}[title={Throughput in DIL networks with different choices of $p_N$},legend style={yshift=-2.5cm, xshift=0.05cm},\axissetting]
                                      
\pfgsetting
\addplot+[color=blue,mark=none] table [x=N,y=1/N]{RRPIL.dat};\addlegendentry{$p_N=\frac{1}{N}$};
\addplot+[color=magenta,mark=none,dashed] table [x=N,y=1/4N]{RRPIL.dat};\addlegendentry{$p_N=\frac{1}{4N}$};
\addplot+[color=black,mark=none,dotted] table [x=N,y=1/10N]{RRPIL.dat};\addlegendentry{$p_N=\frac{1}{10N}$};
\end{axis}
\end{tikzpicture}
\label{DILPR}}
\centering
\caption{Normalized throughput in DIL networks as a function of number of SUs (a). Secondary network throughput as a function of the number of SUs for different choices of $p_N$ (b). $Q_{\rm ave}$ is set to 0dB.}
\end{figure}

In Fig. \ref{RWeiILR}, we plot the normalized throughputs in DIL networks as a function of the number of SUs to further illustrate the accuracy of our scaling results. In this figure, direct channel gains are Rayleigh distributed and interference channel gains are Weibull-$c$ distributed with $c=1.5, 2.5$. As Fig. \ref{RWeiILR} shows, the sum-rate in DIL networks scales according to $\frac{2}{c\e{}}\logp{N}$, which is in harmony with Theorem \ref{Theo: DIL}.  

Figure \ref{DILPR} depicts the throughput scaling behavior of DIL networks for different selections of the transmission probability. In this figure, direct and interference channel gains are Rayleigh distributed, and $p_N$ is set to $\frac{1}{N}, \frac{1}{4N}$ and $\frac{1}{10N}$. We observe that the sum-rate performance of DIL networks under $p_N=\frac{1}{N}$ is asymptotically much higher compared to that of DIL networks under $p_N=\frac{1}{4N}$ and $p_N=\frac{1}{10N}$.  Similar to DTPIL networks, this observation indicates the importance of correct calibration of $p_N$ to maximize secondary network sum-rates in DIL networks.  

\section{Conclusions}\label{sec: conclusion}
In this paper, we have studied the optimum distributed power control problem and the throughput scaling laws for the distributed cognitive uplink. First, we have shown that the optimum distributed power control policy for the cognitive uplink is in the form of a threshold based water-filling power control. The derived optimum distributed power control policy maximizes the secondary network sum-rate subject to transmission and interference power limitations, whilst guaranteeing primary QoS requirements without any feedback signals. Second, we have derived tight throughput scaling laws for the distributed cognitive uplink by considering fading models general enough to include Rayleigh, Rician and Nakagami fading as special cases.  In particular, it has been shown that the secondary network sum-rate, under the optimum distributed power control policy, scales according to $\frac{1}{\e{}n_h}\log\logp{N}$ when transmission powers of SUs are limited by a total average transmission power constraint and a constraint on the average interference power of SUs at the PBS. Here, $n_h$ is a parameter obtained from the distribution of direct channel power gains, and $N$ is the number of SUs. It has also been shown that the secondary network sum-rate, under the optimum distributed power control policy, scales 	according to $\frac{1}{\e{}\gamma_g}\logp{N}$ when transmission powers of SUs are only limited by an average interference power constraint. Here, $\gamma_g$ is a parameter obtained from the distribution of interference channel power gains. Our throughput scaling results demonstrate that the cognitive uplink operating according to the derived optimum distributed power control policy is able to harvest multiuser diversity gains, even in a distributed fashion without any feedback between SUs and the SBS.  The pre-log multiplier $\frac{1}{\e{}}$ is the cost of distributed implementation of the cognitive uplink. 
\appendices
\section{Proof of Theorem \ref{Theo: OPA-TPIL}}\label{App: OPA-TPIL} 
To prove Theorem \ref{Theo: OPA-TPIL}, we form a new functional optimization problem as follows 
\begin{eqnarray}\label{OP-Auxiliary-2}
\begin{array}{ll}
\underset{\tilde{P}\paren{h,g},W\paren{h,g}}{\mbox{maximize}} & \EW_{h, g}\sqparen{W\paren{h,g} \log\paren{ 1+  h\tilde{P}\paren{h,g}}}\\
\mbox{subject to} & \EW_{h,g}\sqparen{W\paren{h,g}\tilde{P}\paren{h,g}} \leq \frac{P_{\rm ave}}{N} \\
			    & \EW_{h,g} \sqparen{W\paren{h,g}g\tilde{P}\paren{h,g}} \leq \frac{Q_{\rm ave}}{N}\\
			    &\EW_{h,g}\sqparen{W\paren{h,g}}= \frac{1}{N}\\
			    &0\leq W\paren{h,g}\leq 1
\end{array},
\end{eqnarray}
where $\tilde{P}\paren{h,g}$ is a mapping from $\R_+^2$ to $\R_+$. For $\tilde{P}\paren{h,g}=P\paren{h,g}$ and $W\paren{h,g}=\I{P\paren{h,g}>0}$, the optimization problem in \eqref{OP-Auxiliary-2} reduces to the one in \eqref{OP}. Thus, the optimal value of \eqref{OP-Auxiliary-2} serves as an upper bound for the optimal value of \eqref{OP}. Later, we show that this upper bound is achievable for $N$ large enough. Using the change of variable $\Pi\paren{h,g} = \tilde{P}\paren{h,g} W\paren{h,g}$, \eqref{OP-Auxiliary-2} can be transformed into the following convex optimization problem:
\begin{eqnarray}\label{OP-Auxiliary-3}
\begin{array}{ll}
\underset{\Pi\paren{h,g},W\paren{h,g}}{\mbox{maximize}} & \EW_{h, g}\sqparen{W\paren{h,g} \log\paren{ 1+  \frac{h\Pi\paren{h,g}}{W\paren{h,g}}}}\\
\mbox{subject to} & \EW_{h,g}\sqparen{\Pi\paren{h,g}} \leq \frac{P_{\rm ave}}{N} \\
			    & \EW_{h,g} \sqparen{g\Pi\paren{h,g}} \leq \frac{Q_{\rm ave}}{N}\\
			    &\EW_{h,g}\sqparen{W\paren{h,g}}= \frac{1}{N}\\
			    &0\leq W\paren{h,g}\leq 1
\end{array}, 
\end{eqnarray}
It can be shown that the objective function in \eqref{OP-Auxiliary-3} as a function of $\Pi$ and $W$ is concave on $\Rp^2$. The Lagrangian for \eqref{OP-Auxiliary-3} can be written as
\begin{eqnarray}
L\paren{\Pi, W, \lambda_N,\mu_N,\eta_N}= W\paren{h,g} \log\paren{ 1+  \frac{h\Pi\paren{h,g}}{W\paren{h,g}}}-\lambda_N \Pi\paren{h,g}- \mu_N g\Pi\paren{h,g} -\eta_N W\paren{h,g}\nonumber
\end{eqnarray}
where $\lambda_N\geq0$, $\mu_N\geq0$ and $\eta_N$ are Lagrange multipliers associated with the average transmit power, average interference power and  transmission probability constraints, respectively.  Let $\Pi^\star\paren{h,g}$ and $W^\star\paren{h,g}$ be the solutions of \eqref{OP-Auxiliary-3}. Using generalized Karush-Kuhn-Tucker (KKT) conditions \cite{Luenberger68, Boyd}, we have 
\begin{eqnarray}
\frac{\partial L\paren{\Pi, W^\star, \lambda_N,\mu_N,\eta_N}}{\partial \Pi\paren{h,g}} \Big|_{\Pi = \Pi^\star}
=\frac{h}{1+\frac{h\Pi^\star\paren{h,g}}{W^\star\paren{h,g}}}-\lambda_N -\mu_N \left\{
\begin{array}{cc}
 =0 &  \Pi^\star\paren{h,g}>0  \nonumber\\
 \leq 0 &  \Pi^\star\paren{h,g}=0 
\end{array},
\right. 
\end{eqnarray}
which implies $\tilde{P}^\star\paren{h,g}=\paren{\frac{1}{\lambda_N+\mu_N g}-\frac{1}{h}}^+$. From KKT conditions, we also need to have
\begin{eqnarray}
\lefteqn{\frac{\partial L\paren{\Pi^\star, W, \lambda_N,\mu_N,\eta_N}}{\partial W\paren{h,g}} \Big|_{W = W^\star}}\hspace{16cm}\nonumber\\ \lefteqn{=\logp{1+h\tilde{P}^\star\paren{h,g}}-\lambda_N \tilde{P}^\star\paren{h,g} - \mu_N g\tilde{P}^\star\paren{h,g}-\eta_N
\left\{
\begin{array}{cc}
 =0 &  0<W^\star\paren{h,g}<1  \nonumber\\
 \leq 0 &  W^\star\paren{h,g}=0 \nonumber\\
  \geq 0 &  W^\star\paren{h,g}=1 
\end{array},
\right.}\hspace{14cm}
\end{eqnarray}
For $\frac{\partial L\paren{\Pi^\star, W, \lambda_N,\mu_N,\eta_N}}{\partial W\paren{h,g}}=0$, we have $\logp{1+h\tilde{P}^\star\paren{h,g}}-\lambda_N \tilde{P}^\star\paren{h,g} -\mu_N g\tilde{P}^\star\paren{h,g}=\eta_N$, which happens with zero probability since fading channel gains have continuous distributions. Thus, $W^\star\paren{h,g}\in\left\{0,1\right\}$ with probability one. For $\frac{\partial L\paren{\Pi^\star, W, \lambda_N,\mu_N,\eta_N}}{\partial W\paren{h,g}}\geq0$, we have
\begin{eqnarray}\label{Eq: A}
\logp{1+h\tilde{P}^\star\paren{h,g}}-\lambda_N \tilde{P}^\star\paren{h,g} -\mu_N g\tilde{P}^\star\paren{h,g}-\eta_N\geq 0.
\end{eqnarray}
Substituting $\tilde{P}^\star\paren{h,g}$ in \eqref{Eq: A}, we have  
\begin{eqnarray}\label{Eq: B}
\paren{\logp{\frac{h}{\lambda_N+\mu_N g}}+\frac{\lambda_N+\mu_N g}{h}-1}\I{\frac{h}{\lambda_N+\mu_N g}\geq 1}\geq \eta_N.
\end{eqnarray}

Since $G\paren{x}=\logp{x}+\frac{1}{x}-1$ is monotonically increasing for $x \geq 1$, \eqref{Eq: B} implies that $W^\star\paren{h,g}$ can be chosen as $W^\star\paren{h,g}=\I{\frac{h}{\lambda_N+\mu_N g}\geq \ICRDTPIL{\frac{1}{N}}}$. It can be shown that $\lambda_N\leq \frac{1}{P_{\rm ave}}$ and $\mu_N\leq \frac{1}{Q_{\rm ave}}$. Hence, we have $\ICRDTPIL{\frac{1}{N}}\geq F^{-1}_{\frac{1}{P_{\rm ave}},\frac{1}{Q_{\rm ave}}}\paren{1-\frac{1}{N}}\geq 1$ for $N$ large enough. This implies that $P_{\rm DTPIL}\paren{h,g}=\tilde{P}^\star\paren{h,g}W^\star\paren{h,g}$ is also a feasible solution for \eqref{OP} when $N$ is large enough. For $P_{\rm DTPIL}\paren{h,g}=\tilde{P}^\star\paren{h,g}W^\star\paren{h,g}$, the value of objective function in \eqref{OP} is equal to the optimal value of \eqref{OP-Auxiliary-2}, which completes the proof.
\section{Throughput Scaling in DTPIL Networks}\label{app: TPIL}
In this appendix, we first establish some preliminary results.  Then, we use these results to prove Theorem \ref{Theo: DTPIL}.  Lemma \ref{Lem: inv-scaling} below establishes the asymptotic behavior of $F^{-1}_{\lambda,\mu}\paren{x}$, which is the {\em functional} inverse of the common CDF of joint channel states $\frac{h_i}{\lambda+\mu g_i}$, $i = 1, 2, \ldots, N$, as $x$ becomes close to one.  
\begin{lemma}\label{Lem: inv-scaling}
Let $F_{\lambda,\mu}\paren{x}$ be the common CDF of joint channel states $\frac{h_i}{\lambda+\mu g_i}$, $i = 1, 2, \ldots, N$, where $\lambda > 0$ and $\mu \geq 0$ are constants. Then, as $x$ becomes close to one, its functional inverse $F^{-1}_{\lambda,\mu}\paren{x}$ scales according to
\begin{eqnarray}
\lim_{x\uparrow1}\frac{F^{-1}_{\lambda,\mu}\paren{x}}{\frac{1}{\lambda}\paren{-\frac{1}{\beta_h}\logp{1-x}}^\frac{1}{n_h}}=1.\nonumber
\end{eqnarray}
\end{lemma}
\begin{IEEEproof}
We only focus on the case where both $\lambda$ and $\mu$ are strictly positive.  The proof of the remaining case in which $\lambda >0$ and $\mu = 0$ is easier and follows from the same lines.  To prove the desired result, we first obtain the asymptotic behavior of $F_{\lambda,\mu}\paren{x}$ as $x$ becomes large.  Note that $F_{\lambda,\mu}\paren{x}$ is the CDF of the product of two independent random variables, {\em i.e.,} $h_i$ and $\frac{1}{\lambda + \mu g_i}$, and the asymptotic tail behavior for the product of two independent random variables has been studied in \cite{Andrey} for the case of $H\paren{x} = 0$. Since $H\paren{x}$ is not necessarily equal to zero for the class $\mathcal{C}$-distributions, {\em i.e.,} see the Rician fading model in Table \ref{Table: FP}, we need to upper and lower bound the tail of $F_{\lambda,\mu}\paren{x}$ by using distribution functions with $H\paren{x} = 0$.  To this end, we let $h_{+\epsilon}$ and $h_{-\epsilon}$ be two random variables, independent of $g_i$, with respective CDFs $F_{+\epsilon}(x)$ and  $F_{-\epsilon}(x)$ satisfying 
$\label{Eq: tail-condition-1}
 \lim_{x\ra\infty}\frac{1-F_{+\epsilon}(x)}{\alpha_h x^{l_h}\e{-\paren{\beta_h-\epsilon} x^{n_h}}}=\lim_{x\ra\infty}\frac{1-F_{-\epsilon}(x)}{\alpha_h x^{l_h}\e{-\paren{\beta_h+\epsilon} x^{n_h}}}=1 \nonumber
$
for $\epsilon > 0$ small enough. 

Let $F_{+\epsilon,\lambda,\mu}\paren{x}$ and $F_{-\epsilon,\lambda,\mu}\paren{x}$ be the CDFs of $\frac{h_{+\epsilon}}{\lambda+\mu g_i}$ and $\frac{h_{-\epsilon}}{\lambda+\mu g_i}$, respectively.  Let also $F_{h}(x)$ be the CDF of $h_i$.  Observing that $F_{+\epsilon}\paren{x}\leq F_{h}\paren{x}\leq F_{-\epsilon}\paren{x}$ for $x$ large enough, we can upper and lower bound $F_{\lambda,\mu}(x)$ as 
\begin{eqnarray}
F_{+\epsilon,\lambda,\mu}\paren{x} = \EW\sqparen{F_{+\epsilon}\paren{\paren{\lambda+\mu g_i}x}} \leq F_{\lambda,\mu}\paren{x} = \EW\sqparen{F_{h}\paren{\paren{\lambda+\mu g_i}x}}\leq F_{-\epsilon,\lambda,\mu}\paren{x} = \EW\sqparen{F_{-\epsilon}\paren{\paren{\lambda+\mu g_i}x}} \label{Eqn: CDF Bounds}
\end{eqnarray}
for $x$ large enough, where expectations are taken over interference channel states.  Using Theorem 3 in \cite{Andrey}, the asymptotic tail behavior of $F_{+\epsilon, \lambda,\mu}\paren{x}$ can be shown to satisfy $\lim_{x\ra\infty}\frac{1-F_{+\epsilon,\lambda,\mu}\paren{x}}{Cx^{l_h-n_h\gamma_g}\e{-\paren{\beta_h-\epsilon}\paren{\lambda x}^{n_h}}}=1$, where $C=\eta_g\alpha_h\Gamma\paren{\gamma_g+1}\paren{\frac{\lambda^2}{\mu\paren{\beta_h-\epsilon} {n_h}}}^{\gamma_g}\paren{\frac{1}{\lambda}}^{n_h\gamma_g+\gamma_g-l_h}$ and $\Gamma\paren{\cdot}$ is the Gamma function.  This result implies that the functional inverse $F^{-1}_{+\epsilon,\lambda,\mu}\paren{x}$ of $F_{+\epsilon, \lambda,\mu}\paren{x}$ behaves according to $\lim_{x\uparrow1}\frac{F^{-1}_{+\epsilon,\lambda,\mu}\paren{x}}{\frac{1}{\lambda}\paren{-\frac{1}{\paren{\beta_h-\epsilon}}\logp{1-x}}^\frac{1}{n_h}}=1$ as $x$ becomes close to one.  Following the same steps, we also have $\lim_{x\uparrow1}\frac{F^{-1}_{-\epsilon,\lambda,\mu}\paren{x}}{\frac{1}{\lambda}\paren{-\frac{1}{\paren{\beta_h+\epsilon}}\logp{1-x}}^\frac{1}{n_h}}=1$.  Using \eqref{Eqn: CDF Bounds}, $F^{-1}_{\lambda,\mu}\paren{x}$ can be upper and lower bounded as $F^{-1}_{-\epsilon,\lambda,\mu}\paren{x}\leq F^{-1}_{\lambda,\mu}(x)\leq F^{-1}_{+\epsilon,\lambda,\mu}\paren{x}$ for $x$ close enough to one.  Since $\epsilon$ can be chosen arbitrarily close to zero, we have
\begin{eqnarray}
\lim_{x\uparrow1}\frac{F^{-1}_{\lambda,\mu}\paren{x}}{\frac{1}{\lambda}\paren{-\frac{1}{\beta_h}\logp{1-x}}^\frac{1}{n_h}}=1,\nonumber
\end{eqnarray}
which completes the proof.   
\end{IEEEproof}

Next, by using Lemma \ref{Lem: inv-scaling}, we establish the asymptotic behavior for the extreme order statistic of the collection of random variables $\brparen{\frac{h_i}{\lambda+\mu g_i}}_{i=1}^N$. The derived convergence behavior will be helpful for studying the asymptotic behavior of $\lambda_N$, and in turn, for proving Theorem \ref{Theo: DTPIL}.
\begin{lemma}\label{Lem: EOS-Conc}
Let $X^\star_N\paren{\lambda,\mu}=\max_{1\leq i\leq N}\frac{h_i}{\lambda+\mu g_i}$ for $\lambda > 0$ and $\mu \geq 0$.  Then, $\frac{X^\star_N\paren{\lambda,\mu}}{{\paren{\frac{1}{\beta_h}\logp{N}}^\frac{1}{n_h}}} \xrightarrow{i.p.}\frac{1}{\lambda}$ as $N$ tends to infinity, where $i.p.$ stands for convergence in probability.
\end{lemma}
\begin{IEEEproof}
Let $F_{\lambda,\mu}\paren{x}$ be the CDF of $\frac{h_i}{\lambda+\mu g_i}$ as in Lemma \ref{Lem: inv-scaling}. Using Lemma 2 in \cite{NIDSubmitted}, the concentration behavior of $X^\star_N\paren{\lambda, \mu}$ can be given as  
\begin{eqnarray}
\lim_{N \ra \infty} \PR{F^{-1}_{\lambda,\mu}\paren{1 - N^{\epsilon - 1}} \leq X^\star_N\paren{\lambda, \mu} \leq F^{-1}_{\lambda,\mu}\paren{1 - N^{-\epsilon - 1}}} = 1 \label{Eqn: Extreme Order Scaling}
\end{eqnarray}
for all $\epsilon > 0$ small enough. 
Using Lemma \ref{Lem: inv-scaling} above and \eqref{Eqn: Extreme Order Scaling}, we have 
\begin{eqnarray}
\lim_{N \ra \infty} \PR{ \frac{1}{\lambda}\paren{1-\epsilon}^{\frac{1}{n_h}} \leq \frac{X^\star_N\paren{\lambda, \mu}}{{\paren{\frac{1}{\beta_h}\logp{N}}^\frac{1}{n_h}}} \leq \frac{1}{\lambda}\paren{1+ \epsilon}^{\frac{1}{n_h}}} = 1,\nonumber
\end{eqnarray}
which implies the convergence of $\frac{X^\star_N\paren{\lambda,\mu}}{{\paren{\frac{1}{\beta_h}\logp{N}}^\frac{1}{n_h}}}$ to $\frac{1}{\lambda}$ in probability. 
\end{IEEEproof}

In the next lemma, we show that $\lambda_N$ converges to $\frac{1}{P_{\rm ave}}$ as $N$ becomes large.  This lemma will be used to quantify the effect of the average total power constraint $P_{\rm ave}$ on the secondary network throughput in DTPIL networks.
\begin{lemma}\label{Lem: lambda-conv}
Let $\lambda_N$ be the power control parameter in DTPIL networks. Then, $\lim_{N\ra\infty}\lambda_N=\frac{1}{P_{\rm ave}}$.
\end{lemma}
\begin{IEEEproof}
First, we show $\liminf_{N\ra\infty}\lambda_N>0$.  To obtain a contradiction, we assume that $\lambda_N$ can be arbitrarily close to zero as $N$ becomes large.  This implies that for all $\epsilon>0$, we can find a subsequence of $N$, $N_j$, such that $\lambda_{N_j} \leq \epsilon$ for all $N_j$ large enough.  Let $X^\star_N\paren{\lambda, \mu} = \max_{1 \leq i \leq N} \frac{h_i}{\lambda + \mu g_i}$ (as in Lemma \ref{Lem: EOS-Conc}), $h^\star_N=\max_{1\leq i\leq N}h_i$ and $I_N=\arg\max_{1\leq i \leq N}\frac{h_i}{\lambda_{N}+\mu_{N}g_i}$.  Let also $P_{N}^{\rm DTPIL}\paren{\vec{h},\vec{g}}=\sum_{i=1}^{N}\paren{\frac{1}{\lambda_N+\mu_Ng_i}-\frac{1}{h_i}}^+\I{\CRTPIL>\ICRDTPIL{\frac{1}{N}}}$ be the instantaneous total power consumed by the secondary network. Then,  the average power consumption, for all $N_j$ large enough, can be lower bounded as 
\begin{eqnarray}\label{Eq: LB-TP}
\ES{P_{N_j}^{\rm DTPIL}\paren{\vec{h},\vec{g}}} &=& \ES{\sum_{i=1}^{N_j}\paren{\frac{1}{\lambda_{N_j}+\mu_{N_j}g_i}-\frac{1}{h_i}}^+\I{\frac{h_i}{\lambda_{N_j}+\mu_{N_j}g_i}>F^{-1}_{\lambda_{N_j},\mu_{N_j}}\paren{1-\frac{1}{N_j}}}}\nonumber\\
&\geq& \ES{\frac{1}{h_{I_{N_j}}}\paren{X^\star_{N_j}\paren{\lambda_{N_j},\mu_{N_j}}-1}^+\I{X^\star_{N_j}\paren{\lambda_{N_j},\mu_{N_j}}>F^{-1}_{\lambda_{N_j},\mu_{N_j}}\paren{1-\frac{1}{N_j}}}}\nonumber\\
&\stackrel{(a)}{\geq}&\ES{\frac{1}{h^\star_{N_j}}\paren{X^\star_{N_j}\paren{\epsilon,\frac{1}{Q_{\rm ave}}}-1}^+\I{X^\star_{N_j}\paren{\lambda_{N_j},\mu_{N_j}}>F^{-1}_{\lambda_{N_j},\mu_{N_j}}\paren{1-\frac{1}{N_j}}}},
\end{eqnarray}
where $(a)$ follows from observing that $\mu_N\leq \frac{N p_N}{Q_{\rm ave}}$ and $p_N = \frac{1}{N}$ in this case.  Using Lemma \ref{Lem: EOS-Conc}, we have $\frac{X^\star_{N_j}\paren{\epsilon,\frac{1}{Q_{\rm ave}}}}{{\paren{\frac{1}{\beta_h}\logp{N_j}}^\frac{1}{n_h}}}\xrightarrow{i.p.}\frac{1}{\epsilon}$ and $\frac{h^\star_{N_j}}{\paren{\frac{1}{\beta_h}\logp{N_j}}^\frac{1}{n_h}}\xrightarrow{i.p.}1$ as $N_j$ tends to infinity.  

Also, it is easy to see that 
$
\I{X^\star_{N_j}\paren{\lambda_{N_j},\mu_{N_j}}>F^{-1}_{\lambda_{N_j},\mu_{N_j}}\paren{1-\frac{1}{N_j}}}\xrightarrow{i.d.}\mbox{Bern}\paren{1-\frac{1}{\e{}}}
$
as $N_j$ tends to infinity, where $\mbox{Bern}\paren{p}$ denotes a 0-1 Bernoulli random variable with mean $p$, and $i.d.$ stands for convergence in distribution.  Hence, by using Slutsky's Theorem \cite{Gut}, we have 
 \begin{eqnarray}
 \frac{1}{h^\star_{N_j}}\paren{X^\star_{N_j}\paren{\epsilon,\frac{1}{Q_{\rm ave}}}-1}^+\I{X^\star_{N_j}\paren{\lambda_{N_j},\mu_{N_j}}>F^{-1}_{\lambda_{N_j},\mu_{N_j}}\paren{1-\frac{1}{N_j}}}\xrightarrow{i.d.} \frac{1}{\epsilon} \mbox{Bern}\paren{1-\frac{1}{\e{}}}. \nonumber
 \end{eqnarray} 

Applying Fatou's Lemma to \eqref{Eq: LB-TP}, we obtain $\liminf_{N_j \ra \infty}\ES{P_{N_j}^{\rm DTPIL}\paren{\vec{h},\vec{g}}}\geq \frac{1}{\epsilon}\paren{1-\frac{1}{\e{}}}$, which implies that the average power consumption can be made arbitrarily large, violating the power constraint, for $\epsilon$ small enough and $N_j$ large enough. Thus, $\liminf_{N\ra\infty}\lambda_N>0$. 

Now, by using the fact that $\lambda_N$ cannot be arbitrarily close to zero, we show that $\lim_{N\ra\infty}\lambda_N=\frac{1}{P_{\rm ave}}$. Note that $\lambda_N\leq \frac{1}{P_{\rm ave}}$ for $p_N = \frac{1}{N}$, which implies that $\limsup_{N\ra\infty}\lambda_N\leq\frac{1}{P_{\rm ave}}$.  Hence, showing that $\liminf_{N \ra \infty}\lambda_N \geq \frac{1}{P_{\rm ave}}$ will conclude the proof.  To this end, the average total power consumed by the secondary network can be lower bounded as 
\begin{eqnarray}\label{Eq: Power-lb}
P_{\rm ave}&=&\ES{\sum_{i=1}^{N}\paren{\frac{1}{\lambda_N+\mu_Ng_i}-\frac{1}{h_i}}^+\I{\CRTPIL>\ICRDTPIL{\frac{1}{N}}}}\nonumber\\
&\stackrel{(a)}{=}& \frac{1}{\lambda_N}\ES{\sum_{i=1}^{N}\paren{\frac{1}{1+\frac{\mu_N}{\lambda_N}g_i}-\frac{\lambda_N}{h_i}}^+\I{\frac{h_i}{1+\frac{\mu_N}{\lambda_N}g_i}>F^{-1}_{1,\frac{\mu_N}{\lambda_N}}\paren{1-\frac{1}{N}}}}\nonumber\\
&\stackrel{(b)}{\geq}& \frac{1}{\lambda_N}\ES{\sum_{i=1}^{N}\frac{\paren{F^{-1}_{1,\frac{1}{\lambda_NQ_{\rm ave}}}\paren{1-\frac{1}{N}}-\lambda_N}^+}{h_i}\I{\frac{h_i}{1+\frac{\mu_N}{\lambda_N}g_i}>F^{-1}_{1,\frac{\mu_N}{\lambda_N}}\paren{1-\frac{1}{N}}}}\nonumber\\
&\geq& \frac{1}{\lambda_N} \ES{\frac{\paren{F^{-1}_{1,\frac{1}{\lambda_NQ_{\rm ave}}}\paren{1-\frac{1}{N}}-\lambda_N}^+}{h^\star_N}\sum_{i=1}^{N}\I{\frac{h_i}{1+\frac{\mu_N}{\lambda_N}g_i}>F^{-1}_{1,\frac{\mu_N}{\lambda_N}}\paren{1-\frac{1}{N}}}},
\end{eqnarray}
where $(a)$ follows from observing that $\lambda F_{\lambda,\mu}^{-1}\paren{x}=F_{1,\frac{\mu}{\lambda}}^{-1}\paren{x}$, and $(b)$ follows from observing that $\mu_N \leq \frac{1}{Q_{\rm ave}}$ and $F^{-1}_{\lambda, \mu}\paren{x}$ decreases with increasing values of $\mu$.  Using \eqref{Eq: Power-lb}, $\lambda_N$ can be lower bounded as 
\begin{eqnarray}\label{Eq: lambda-lb}
\lambda_N &\geq& \frac{1}{P_{\rm ave}}\ES{\frac{\paren{F^{-1}_{1,\frac{1}{\lambda_NQ_{\rm ave}}}\paren{1-\frac{1}{N}}-\lambda_N}^+}{h^\star_N}\sum_{i=1}^{N}\I{\frac{h_i}{1+\frac{\mu_N}{\lambda_N}g_i}>F^{-1}_{1,\frac{\mu_N}{\lambda_N}}\paren{1-\frac{1}{N}}}}.
\end{eqnarray}

Using Lemma \ref{Lem: inv-scaling} and the fact that $\lambda_N$ cannot be arbitrarily close to zero, we have $\lim_{N\ra\infty}\frac{F^{-1}_{1,\frac{1}{\lambda_NQ_{\rm ave}}}\paren{1-\frac{1}{N}}}{\paren{\frac{1}{\beta_h}\logp{N}}^\frac{1}{n_h}}=1$, which implies $\frac{ \paren{F^{-1}_{1,\frac{1}{\lambda_NQ_{\rm ave}}}\paren{1-\frac{1}{N}}-\lambda_N}^+}{h^\star_N}\xrightarrow{i.p.}1$ as $N$ tends to infinity.  Let $S_N=\sum_{i=1}^{N}\I{\frac{h_i}{1+\frac{\mu_N}{\lambda_N}g_i}>F^{-1}_{1,\frac{\mu_N}{\lambda_N}}\paren{1-\frac{1}{N}}}$.  $S_N$ has a Binomial distribution with parameters $N$ and $\frac{1}{N}$.  Hence, using Poisson approximation for Binomial distributions, we conclude that $S_N$ converges in distribution to $\mbox{Po}\paren{1}$, where $\mbox{Po}\paren{p}$ represents a Poisson random variable with mean $p$. Using Slutsky's Theorem, we have 
\begin{eqnarray}
\frac{\paren{F^{-1}_{1,\frac{1}{\lambda_NQ_{\rm ave}}}\paren{1-\frac{1}{N}}-\lambda_N}^+}{h^\star_N}\sum_{i=1}^{N}\I{\frac{h_i}{1+\frac{\mu_N}{\lambda_N}g_i}>F^{-1}_{1,\frac{\mu_N}{\lambda_N}}\paren{1-\frac{1}{N}}}\xrightarrow{i.d.}\mbox{Po}\paren{1}\nonumber
\end{eqnarray}
as $N$ grows large. 
 Applying Fatou's Lemma to \eqref{Eq: lambda-lb}, we have $\liminf_{N\ra\infty}\lambda_N\geq \frac{1}{P_{\rm ave}}$. 
\end{IEEEproof}

Now, we are ready to prove Theorem \ref{Theo: DTPIL} by utilizing above auxiliary results.  Note that the sum-rate under the optimum distributed power control in DTPIL networks for $p_N = \frac{1}{N}$ can be written as 
\begin{eqnarray}
\RTPIL{\frac{1}{N}}=\logp{\frac{1}{\lambda_N}}\PRP{A_N}+\ES{\logp{\Xs{1}{\frac{\mu_N}{\lambda_N}}}\Inb{A_N}}.\nonumber
\end{eqnarray}
It is easy to see that $\lim_{N\ra\infty}\PRP{A_N}=\frac{1}{\e{}}$ by the selection of transmission probabilities.  This gives us the logarithmic effect of $P_{\rm ave}$ on the secondary network throughput since $\lambda_N$ converges to $\frac{1}{P_{\rm ave}}$.  Using Lemma \ref{Lem: EOS-Conc}, we have $\frac{\logp{X^\star_N\paren{1,\frac{\mu_N}{\lambda_N}}}}{\log\logp{N}}\xrightarrow{i.p.}\frac{1}{n_h}$ as $N$ tends to infinity since $\lambda_N$ is bounded away from zero and $\mu_N \leq \frac{1}{Q_{\rm ave}}$.  Also, we have $\Inb{A_N}$ converging in distribution to $\mbox{Bern}\paren{\frac{1}{\e{}}}$ as $N$ tends to infinity.  As a result, applying Slutsky's Theorem, we conclude that $\frac{\logp{X^\star_N\paren{1,\frac{\mu_N}{\lambda_N}}}}{\log\logp{N}}\Inb{A_N}\xrightarrow{i.d.}\frac{1}{n_h}\mbox{Bern}\paren{\frac{1}{\e{}}}$.  This final result almost completes the proof of Theorem \ref{Theo: DTPIL} up to a slight technicality.  That is, convergence in distribution does not always imply convergence in mean \cite{Billingsley}.   

To show that convergence in mean does also hold in our case, we let $\hat{X}_N\paren{1,\frac{\mu_N}{\lambda_N}}=\frac{\logp{X^\star_N\paren{1,\frac{\mu_N}{\lambda_N}}}}{\log\logp{N}}\Inb{A_N}$.  It is enough to show that the collection of random variables $\brparen{\hat{X}_N\paren{1,\frac{\mu_N}{\lambda_N}}}_{N=1}^\infty$ is uniformly integrable, \emph{i.e.,} 
$\lim_{C^\prime\ra\infty}\sup_{N \geq 1}\ES{\abs{\hat{X}_N\paren{1,\frac{\mu_N}{\lambda_N}}}\I{\abs{\hat{X}_N\paren{1,\frac{\mu_N}{\lambda_N}}}\geq C^\prime}}=0$ to conclude the proof.  We can upper bound the random variable $\frac{\logp{X^\star_N\paren{1,\frac{\mu_N}{\lambda_N}}}}{\log\logp{N}}\Inb{A_N}$ as  $\frac{\logp{X^\star_N\paren{1,\frac{\mu_N}{\lambda_N}}}}{\log\logp{N}}\Inb{A_N}\leq \frac{\logp{X^\star_N\paren{1,\frac{\mu_N}{\lambda_N}}}}{\log\logp{N}}\I{X^\star_N\paren{1,\frac{\mu_N}{\lambda_N}}\geq 1}$.  Using proof techniques similar to those used in the proof of Lemma 3 in \cite{NIDSubmitted}, it can be shown that $\brparen{\frac{\logp{X^\star_N\paren{1,\frac{\mu_N}{\lambda_N}}}}{\log\logp{N}}\I{X^\star_N\paren{1,\frac{\mu_N}{\lambda_N}}\geq 1}}_{N=1}^\infty$ is uniformly integrable, which implies the uniform integrability of  $\brparen{\frac{\logp{X^\star_N\paren{1,\frac{\mu_N}{\lambda_N}}}}{\log\logp{N}}\Inb{A_N}}_{N=1}^\infty$.  
\section{Proof of Theorem \ref{Theo: Opt-Scal-DTPIL}}\label{App: Opt-Scal-DTPIL}
Note that $\RTPIL{p^\star_N}\geq \RTPIL{\frac{1}{N}}$. Hence, it is enough to show that $\limsup_{N\ra\infty}\frac{\RTPIL{p^\star_N}}{\log\logp{N}}\leq \frac{1}{\e{}n_h}$. To this end, let $\Xt=\frac{\logp{\Xs{\lambda_N}{\mu_N}}}{\log\logp{N}}$, where $\Xs{\lambda}{\mu}$ is defined as in Lemma \ref{Lem: EOS-Conc}.  For all $\epsilon>0$, we have
\begin{eqnarray}
\frac{\RTPIL{p^\star_N}}{\log\logp{N}} &=& \ES{\Xt\Inb{A_N}\I{\abs{\Xt-\frac{1}{n_h}}>\epsilon}}+\ES{\Xt\Inb{A_N}\I{\abs{\Xt-\frac{1}{n_h}}\leq \epsilon}}\nonumber\\
&\leq&\ES{\Xt\I{\Xs{\lambda_N}{\mu_N}\geq 1}\I{\abs{\Xt-\frac{1}{n_h}}>\epsilon}}+\paren{\frac{1}{n_h}+\epsilon}\PRP{A_N}.\nonumber
\end{eqnarray}

As in the proof of Theorem \ref{Theo: DTPIL}, we have $\Xt\I{\Xs{\lambda_N}{\mu_N}\geq 1}\xrightarrow{i.p.}\frac{1}{n_h}$ and $\I{\abs{\Xt-\frac{1}{n_h}}>\epsilon}\xrightarrow{i.p.}0$ as $N$ tends to infinity.  Hence, $\Xt\I{\Xs{\lambda_N}{\mu_N}\geq 1}\I{\abs{\Xt-\frac{1}{n_h}}>\epsilon}$ converges to zero in probability.  Using techniques similar to those used in the proof of Theorem \ref{Theo: DTPIL}, it can also be shown that the collection of random variables $\brparen{\Xt\I{\Xs{\lambda_N}{\mu_N}\geq 1}\I{\abs{\Xt-\frac{1}{n_h}}>\epsilon}}_{N=1}^\infty$ is uniformly integrable, which implies that $$\lim_{N\ra\infty}\ES{\Xt\I{\Xs{\lambda_N}{\mu_N}\geq 1}\I{\abs{\Xt-\frac{1}{n_h}>}\epsilon}}=0.$$ 

For $N$ large enough, $\PRP{A_N}$ can be upper bounded as 
\begin{eqnarray}
\PRP{A_N} &=& Np^\star_N\paren{1-p^\star_N}^{N-1}\nonumber\\
&\stackrel{(a)}{\leq}&\paren{1-\frac{1}{N}}^{N-1},\nonumber
\end{eqnarray}
where $(a)$ follows from the fact that $Np^\star_N\paren{1-p^\star_N}^{N-1}$ is maximized at $p^\star_N=\frac{1}{N}$. Hence, 
\begin{eqnarray}
\limsup_{N\ra\infty}\frac{\RTPIL{p^\star_N}}{\log\logp{N}}\leq \frac{1}{\e{}n_h} + \frac{\epsilon}{\e{}},\nonumber
\end{eqnarray}
which completes the proof since $\epsilon$ is arbitrary.
\section{Proof of Lemma \ref{Theo: Opt-Pro-DTPIL}}\label{App: Opt-Pro-DTPIL}
Assume that $a$ is a limit point of the sequence $a_N = N p^\star_N$, where $N \geq 1$.  We only consider the case $a \in (0, \infty)$.  For $a = 0$, it can be shown that the probability of successful transmission, and hence the secondary network throughput, goes to zero due to lack of enough transmission attempts.  For $a = \infty$, the probability of successful transmission, and hence the secondary network throughput, goes to zero due to excessive simultaneous transmission attempts.  

Let $N_j$ be a subsequence of $N$ such that $\lim_{N_j\ra\infty}a_{N_j}=a$.  As argued in the proof of Theorem \ref{Theo: DTPIL}, it can be shown that the probability of successful transmission on this subsequence converges to $\frac{a}{\e{a}}$, {\em i.e.,} $\lim_{N_j\ra\infty}\PRP{A_{N_j}}=\frac{a}{\e{a}}$.  Hence, using techniques similar to those employed in the proof of Theorem \ref{Theo: DTPIL}, we can further show that $\lim_{N_j\ra\infty}\frac{R^\star_{\rm DTPIL}\paren{p^\star_{N_j},N_j}}{\log\logp{N_j}}=\frac{a}{\e{a}n_h}$, which is maximized at $a = 1$.  This implies that $p^\star_N$ must be chosen such that $\lim_{N \ra \infty} N p^\star_N = 1$ to obtain optimal secondary network throughput scaling behavior.      
\section{Throughput Scaling in DIL Networks}\label{app: IL}
To obtain the throughput scaling behavior in DIL networks, we will first provide a preliminary lemma establishing the convergence behavior of $\mu_N$. This lemma will also be helpful to study the effect of average total interference power, $Q_{\rm ave}$, on the secondary network throughput in DIL networks.
\begin{lemma}\label{Lem: mu-conv}
Let $\mu_N$ be the power control parameter in DIL networks. Then, $\lim_{N\ra\infty}\mu_N=\frac{1}{Q_{\rm ave}}$.
\end{lemma}
\begin{IEEEproof}
First, we show that $\mu_N$ is upper bounded by $\frac{1}{Q_{\rm ave}}$ for all $N$. To this end, we have
\begin{eqnarray}
Q_{\rm ave} &=& \ES{\sum_{i=1}^N g_i \paren{\frac{1}{\mu_N g_i} - \frac{1}{h_i}}^+\I{\frac{h_i}{g_i} > \ICRDIL{\frac{1}{N}}}} \nonumber \\
&\leq& \ES{\sum_{i=1}^N \frac{1}{\mu_N} \I{\frac{h_i}{g_i} > \ICRDIL{\frac{1}{N}}}} \nonumber \\
&=& \frac{1}{\mu_N}, \nonumber
\end{eqnarray}
which implies that $\mu_N \leq \frac{1}{Q_{\rm ave}}$. Hence, to complete the proof, it is enough to show that $\liminf_{N \ra \infty} \mu_N \geq \frac{1}{Q_{\rm ave}}$.  We can lower bound $\mu_N$ as
\begin{eqnarray}
\mu_N &=& \frac{1}{Q_{\rm ave}} \ES{\sum_{i=1}^N \paren{1 - \frac{g_i \mu_N}{h_i}}^{+} \I{\frac{h_i}{g_i} > \ICRDIL{\frac{1}{N}}}} \nonumber \\
&\geq& \frac{1}{Q_{\rm ave}} - \frac{\mu_N}{Q_{\rm ave}} N \ES{\frac{g_1}{h_1} \I{\frac{h_1}{g_1} > \ICRDIL{\frac{1}{N}}}} \nonumber \\
&\geq& \frac{1}{Q_{\rm ave}} - \frac{\mu_N}{Q_{\rm ave} \ICRDIL{\frac{1}{N}}}. \label{Eqn: mu_N lower bound}
\end{eqnarray}   

Since $\mu_N$ is bounded above by $\frac{1}{Q_{\rm ave}}$ and $\ICRDIL{\frac{1}{N}}$ tends to infinity as $N$ grows large, \eqref{Eqn: mu_N lower bound} implies that $\liminf_{N \ra \infty} \mu_N = \frac{1}{Q_{\rm ave}}$.   
\end{IEEEproof}

Now, we are ready to prove Theorem \ref{Theo: DIL}.  The sum-rate under the optimum distributed power control in DIL networks can be written as 
\begin{eqnarray}\label{Eq: rate-expan}
\RIL{\frac{1}{N}}=\logp{\frac{1}{\mu_N}}\PRP{B_N}+\ES{\logp{\Ys}\Inb{B_N}}.
\end{eqnarray}

It is easy to see that $\lim_{N\ra\infty}\PRP{B_N}=\frac{1}{\e{}}$. Thus, the first term on the right-hand side of \eqref{Eq: rate-expan} converges to $\frac{1}{\e{}}\logp{Q_{\rm ave}}$ as $N$ tends to infinity, which indicates the logarithmic effect of the average interference power constraint, $Q_{\rm ave}$, on the secondary network sum-rate in DIL networks.  It can also be shown that $\frac{\logp{\Ys}}{\logp{N}}\xrightarrow{i.p.}\frac{1}{\gamma_g}$ ({\em i.e.,} see Lemma 8 in \cite{NIDSubmitted}) and $\Inb{B_N} \xrightarrow{i.d.} \mbox{Bern}\paren{\frac{1}{\e{}}}$ as $N$ tends to infinity.  Therefore, using Slutsky's theorem, we have $\frac{\logp{\Ys}}{\logp{N}}\Inb{B_N} \xrightarrow{i.d.} \mbox{Bern}\paren{\frac{1}{\e{}}}$ as $N$ grows large.  Since convergence in distribution does not always imply convergence in mean, we need to show that the collection of random variables $\brparen{\frac{\logp{\Ys}}{\logp{N}}\Inb{B_N}}_{N=1}^\infty$ is uniformly integrable.  For $N$ large enough, we have $\frac{\logp{\Ys}}{\logp{N}}\Inb{B_N}\leq\frac{\logp{\Ys}}{\logp{N}}\I{\Ys\geq 1}$.  Using Lemma 8 in \cite{NIDSubmitted}, we conclude that $\brparen{\frac{\logp{\Ys}}{\logp{N}}\I{\Ys\geq 1}}_{N=1}^\infty$ is uniformly integrable, which implies uniform integrability of $\brparen{\frac{\logp{\Ys}}{\logp{N}}\Inb{B_N}}_{N=1}^\infty$. Hence, we have $\lim_{N \ra \infty}\ES{\frac{\logp{\Ys}}{\logp{N}}\Inb{B_N}}=\frac{1}{\e{}\gamma_g}$, which concludes the proof.

\section{Proof of Theorem \ref{Theo: Opt-Scal-DIL}}\label{App: Opt-Scal-DIL}
Since $\RIL{p^\star_N}\geq \RIL{\frac{1}{N}}$, we have $\liminf_{N \ra \infty}\frac{\RIL{p^\star_N}}{\logp{N}}\geq \frac{1}{\e{}\gamma_g}$. To show the other direction, let $\Yt=\frac{\logp{\frac{\Ys}{\mu_N}}}{\logp{N}}$. For all $\epsilon>0$, we have
\begin{eqnarray}
\frac{\RIL{p^\star_N}}{\logp{N}}&=&\ES{\Yt\Inb{B_N}\I{\abs{\Yt-\frac{1}{\gamma_g}}>\epsilon}}+\ES{\Yt\Inb{B_N}\I{\abs{\Yt-\frac{1}{\gamma_g}}\leq \epsilon}}\nonumber\\
&\leq&\ES{\Yt\I{\Ys\geq \mu_N}\I{\abs{\Yt-\frac{1}{\gamma_g}}>\epsilon}}+\paren{\frac{1}{\gamma_g}+\epsilon}\PRP{B_N}.\nonumber
\end{eqnarray}

Recall from the proof of Theorem \ref{Theo: DIL} that $\frac{Y_N^\star}{\log(N)}$ converges in probability to $\frac{1}{\gamma_g}$.  This implies that $\Yt\I{\Ys\geq \mu_N}\xrightarrow{i.p.}\frac{1}{\gamma_g}$ and $\I{\abs{\Yt-\frac{1}{\gamma_g}}>\epsilon}\xrightarrow{i.p.}0$ as $N$ tends to infinity.  Hence, we have $\Yt\I{\Ys\geq \mu_N}\I{\abs{\Yt-\frac{1}{\gamma_g}}>\epsilon}$ converging in probability to $0$.  Using techniques similar to those used in the proof of Theorem \ref{Theo: DIL}, we can show that the collection of random variables $\brparen{\Yt\I{\Ys\geq \mu_N}\I{\abs{\Yt-\frac{1}{\gamma_g}}>\epsilon}}_{N=1}^\infty$ is uniformly integrable.  This implies that $$\lim_{N\ra\infty}\ES{\Yt\I{\Ys\geq \mu_N}\I{\abs{\Yt-\frac{1}{\gamma_g}}>\epsilon}}=0.$$ 

For $N$ large enough, $\PRP{B_N}$ can be upper bounded as 
\begin{eqnarray}
\PRP{B_N} &=& Np^\star_N\paren{1-p^\star_N}^{N-1}\nonumber\\
&\leq&\paren{1-\frac{1}{N}}^{N-1}.\nonumber
\end{eqnarray}
Hence, 
\begin{eqnarray}
\limsup_{N\ra\infty}\frac{\RIL{p^\star_N}}{\logp{N}}\leq \frac{1}{\e{}\gamma_g} + \frac{\epsilon}{\e{}},\nonumber
\end{eqnarray}
which completes the proof since $\epsilon$ is arbitrary.
\section{Proof of Lemma \ref{Theo: Opt-Pro-DIL}}\label{App: Opt-Pro-DIL}
The proof of this lemma is similar to the proof of Lemma \ref{Theo: Opt-Pro-DTPIL}.  Assume that $a$ is a limit point of the sequence $a_N = N p^\star_N$, $N \geq 1$, and let $N_j$ be a subsequence of $N$ achieving $a$.  For $a = 0$, it can be shown that the probability of successful transmission, and hence the secondary network throughput, goes to zero due to lack of enough transmission attempts.  For $a = \infty$, the probability of successful transmission, and hence the secondary network throughput, goes to zero due to excessive simultaneous transmission attempts.  

For $a \in (0, \infty)$, as argued in the proof of Theorem \ref{Theo: DIL}, it can be shown that the probability of successful transmission on $N_j$ converges to $\frac{a}{\e{a}}$, which, in turn, leads to $\lim_{N_j\ra\infty}\frac{R^\star_{\rm DIL}\paren{p^\star_{N_j},N_j}}{\logp{N_j}}=\frac{a}{\e{a}n_h}$.  Since $\frac{a}{\e{a}n_h}$ is maximized at $a = 1$, we have $\lim_{N \ra \infty} N p^\star_N = 1$. 

\bibliographystyle{IEEEtran}

\begin{thebibliography}{1}
\providecommand{\url}[1]{#1}
\csname url@samestyle\endcsname
\providecommand{\newblock}{\relax}
\providecommand{\bibinfo}[2]{#2}
\providecommand{\BIBentrySTDinterwordspacing}{\spaceskip=0pt\relax}
\providecommand{\BIBentryALTinterwordstretchfactor}{4}
\providecommand{\BIBentryALTinterwordspacing}{\spaceskip=\fontdimen2\font plus
\BIBentryALTinterwordstretchfactor\fontdimen3\font minus
  \fontdimen4\font\relax}
\providecommand{\BIBforeignlanguage}[2]{{%
\expandafter\ifx\csname l@#1\endcsname\relax
\typeout{** WARNING: IEEEtran.bst: No hyphenation pattern has been}%
\typeout{** loaded for the language `#1'. Using the pattern for}%
\typeout{** the default language instead.}%
\else
\language=\csname l@#1\endcsname
\fi
#2}}
\providecommand{\BIBdecl}{\relax}
\BIBdecl

\bibitem{Goldsmith09}
A.~Goldsmith, S. A.~Jafar, I.~Maric and S.~Srinivasa, ``Breaking spectrum gridlock with cognitive radios: An information theoretic perspective,'' \emph{Proc. IEEE}, vol. 97, no. 5, pp. 894-914, May 2009.

\bibitem{FCC02}
FCC Spectrum Policy Task Force, ``Report of the spectrum efficiency working group," {\it Tech. Rep. 02-135}, Nov. 2002. Available: http://www.fcc.gov/sptf/files/SEWGFinalReport$\_$1.pdf.

\bibitem{Akyildiz06} I. F. Akyildiz, W.-Y. Lee, M. C. Vuran and S. Mohanty, ``NeXt generation/dynamic spectrum access/cognitive radio wireless networks: A survey," {\it Computer Networks}, vol. 50, no. 13, pp. 2127-2159, Sept. 2006.

\bibitem{Mitola99b}
J. Mitola III and G. Q. Maguire JR., ``Cognitive radio: Making software radios more personal," {\it IEEE Pers. Commun.}, vol. 6, no. 4, pp. 13-18, Aug. 1999.

\bibitem{Haykin05}
S.~ Haykin, ``Cognitive radio: Brain-empowered wireless communications", \emph{IEEE J. Sel. Areas Commun.}, vol. 23, no. 2, pp. 201-220, Feb. 2005.

\bibitem{Ghasemi07} 
A. Ghasemi and E. S. Sousa, ``Fundamental limits of spectrum-sharing in fading environments,'' {\it IEEE Trans. Wireless Commun.}, vol. 6, no. 2, pp. 649-658, Feb. 2007.

\bibitem{RZhang09}
R. Zhang, S. Cui and Y.-C. Liang, ``On ergodic sum capacity of fading cognitive multiple-access and broadcast channels,'' {\it IEEE Trans. Inf. Theory}, vol. 55, no. 11, pp. 5161-5178, Nov. 2009.

\bibitem{NID12}
E.~Nekouei, H.~Inaltekin and S.~Dey, ``Throughput scaling in cognitive multiple access with average power and interference constraints,'' {\it IEEE Trans. Signal Process.}, vol. 60, no. 2, pp. 927-946, Feb. 2012.

\bibitem{NIDSubmitted}
E.~Nekouei, H.~Inaltekin and S.~Dey, ``Asymptotically optimal feedback protocol design for cognitive multiple access channels," Available online: http://arxiv.org/abs/1209.1424, Tech. Rep., Sept. 2012.

\bibitem{IH12}
H. Inaltekin and S. V. Hanly, ``Optimality of binary power control for the single cell uplink," {\it IEEE Trans. Inf. Theory}, vol. 58, no. 10, pp. 6484-6498, Oct 2012.

\bibitem{cogmud_twban09} 
T. W. Ban, W. Choi, B. C. Jung, D. K. Sung, ``Multi-user diversity in a spectrum sharing system," {\it IEEE Trans. Wireless Commun.}, vol. 8, no. 1, pp. 102-106, Jan. 2009.

\bibitem{cogmid_zhang10}
R.~Zhang and Y.-C.~Liang, ``Investigation on multiuser diversity in spectrum sharing based cognitive radio networks," {\it IEEE Commun. Letters}, vol. 14, no. 2, pp. 133-135, Feb. 2010.


\bibitem{Yang12-TW}
L.~Yang and A.~Nosratinia, ``Hybrid opportunistic scheduling in cognitive radio networks," {\it IEEE Trans. Wireless Commun.}, vol. 11, no. 1, pp.  328-337, Jan. 2012.

\bibitem{Yang12-IT}
L.~Yang and A.~Nosratinia, ``Capacity limits of multiuser multiantenna cognitive networks," {\it IEEE Trans. Inf. Theory}, vol. 58, no. 7, pp. 4493-4508, Jul. 2012.

\bibitem{cogmudmsd_10}
H. Wang, J. Lee, S. Kim and D. Hong, ``Capacity of secondary users exploiting multispectrum and multiuser diversity in spectrum-sharing environments,"  {\it IEEE Trans. Veh. Technol.}, vol. 59, no. 2, pp. 1030-1036, Feb. 2010.

\bibitem{cogmud_alitajer10}
A.~Tajer and X.~Wang, ``Multiuser diversity gain in cognitive networks," {\it IEEE/ACM Trans. Netw.}, vol. 18, no. 6, pp. 1766-1779, Dec. 2010.

\bibitem{Berry06}
X. Qin and R. A. Berry, ``Distributed approaches for exploiting multiuser diversity in wireless networks," {\it IEEE Trans. Inf. Theory}, vol. 52, no. 2, pp. 392-413, Feb. 2006.

\bibitem{Boche07}
H.~Boche, M.~Wiczanowski, ``Optimization-theoretic analysis of stability-optimal transmission policy for multiple-antenna multiple-access channel," {\it IEEE Trans. Signal Process.}, vol. 55, no. 6, pp. 2688-2702, June 2007.

\bibitem{Tse}
D. N. C. Tse and P.~Viswanath., \textit{Fundamentals of Wireless Communication}, Cambridge University Press, Cambridge, 2005.


\bibitem{Simon-Alouini05} M. K.~Simon, M.-S. Alouini, {\em Digital Communication over Fading Channels}, Wiley-Interscience, Hoboken, New Jersey, second edition, 2005.

\bibitem{Stuber96}
G.~St\"uber, {\it Principles of Mobile Communication}, Kluwer Academic Publishers, Boston, 1996.

\bibitem{Bertoni88} H.~Bertoni, ``Coverage prediction for mobile radio systems operating in the 800/900
MHz frequency range," \emph{ IEEE Trans. Veh. Technol.,}
vol. 37, no. 1, pp. 57–60, Feb. 1988.

 


\bibitem{Billingsley}
P.~Billingsley, \textit{Probability and Measure}, New York: John Wiley and Sons, third edition, 1995.

\bibitem{Gut}
A.~Gut, {\em Probability: A Graduate Course}, Springer, New York, 2005. 


\bibitem{Andrey}
A.~Sarantsev, ``Tail Asymptotic of Sum and Product of Random Variables with Applications in the Theory of Extremes of Conditionally Gaussian Processes," \emph{Master Thesis}, Lomonosov Moscow State University, Moscow, July, 2010. Available online: http://arxiv.org/pdf/1107.3869v1.pdf. 


\bibitem{Luenberger68}
D. G. Luenberger, {\it Optimization by Vector Space Methods}, Wiley, New York, 1968.

\bibitem{Boyd}
S. Boyd and L. Vandenberghe, \textit{Convex Optimization}, Cambridge University Press, New York, 2004.




\end{thebibliography}

\end{document}